\documentclass[final,3p,times,twocolumn]{elsarticle} 
\usepackage{amsmath,amsthm,amssymb,amsfonts}
\usepackage{latexsym}
\usepackage{graphicx,psfrag,import}
\usepackage{graphicx}
\usepackage{fullpage}
\usepackage{framed}
\usepackage{verbatim}
\usepackage{color}
\usepackage{epsfig}
\usepackage{geometry}
\usepackage{mathtools}
\usepackage{nonfloat}
\usepackage{enumerate}
\usepackage{multicol}
\usepackage{a4wide}
\usepackage{booktabs}
\usepackage{enumitem}
\usepackage{lineno}
\usepackage{parcolumns}
\usepackage{thmtools}
\usepackage{xr}
\usepackage{epstopdf}
\usepackage{mathrsfs}
\usepackage{subfig}
\usepackage{caption}

\newcommand\myfigure[1]{
\medskip\noindent\begin{minipage}{\columnwidth}
\centering
\end{minipage}\medskip}

\newsavebox{\measurebox}

\theoremstyle{plain}
\newtheorem{theorem}{Theorem}[section]

\newtheorem{lemma}[theorem]{Lemma}

\newtheorem{question}[theorem]{Question}

\theoremstyle{definition}

\date{}

\renewcommand{\figurename}{Figure}
\newcommand{\teo}[1]{{\color{black} #1}}

\begin{document}

\begin{frontmatter}

\title{\textbf{Optimizing enzymatic catalysts for rapid turnover of substrates with low enzyme sequestration}}

\author[1]{Abhishek Deshpande}
\ead{deshpande8@wisc.edu}
\author[2]{Thomas E. Ouldridge\corref{cor2}}
\ead{t.ouldridge.imperial.ac.uk}

\cortext[cor2]{Corresponding author}
\address[1]{Department of Mathematics, University of Wisconin Madison, Madison, WI 53706, United States of America}
\address[2]{Imperial College Centre for Synthetic Biology and Department of Bioengineering, Imperial College London, London, SW7 2AZ, United Kingdom}

\begin{abstract}
Enzymes are central to both metabolism and information processing in cells.  In both cases, an enzyme's ability to accelerate a reaction without being consumed in the reaction is crucial.  Nevertheless, enzymes are transiently sequestered when they bind to their substrates; this sequestration limits activity and potentially compromises information processing and signal transduction. In this article we analyse the mechanism of enzyme-substrate catalysis from the perspective of minimizing the load on the enzymes through sequestration, whilst maintaining at least a minimum reaction flux. In particular, we ask: which binding free energies of the enzyme-substrate and enzyme-product reaction intermediates minimize the fraction of enzymes sequestered in complexes, while sustaining a certain minimal flux? Under reasonable biophysical assumptions, we find that the optimal design will saturate the bound on the minimal flux, and reflects a basic trade-off in catalytic operation. If both binding free energies are too high, there is low sequestration, but the effective progress of the reaction is hampered. If both binding free energies are too low, there is high sequestration, and the reaction flux may also be suppressed in extreme cases. The optimal binding free energies are therefore neither too high nor too low, but in fact moderate. Moreover, the optimal difference in substrate and product binding free energies, which contributes to the thermodynamic driving force of the reaction, is in general strongly constrained by the intrinsic free-energy difference between products and reactants. Both the strategies of using  a negative binding free-energy difference to drive the catalyst-bound reaction forward, and of using a positive binding free-energy difference to enhance detachment of the product, are limited in their efficacy. 
\end{abstract}

\end{frontmatter}

\section{Introduction}
    
Enzymatic catalysts are ubiquitous in biology, forming crucial parts of the networks that implement metabolism~\cite{berg2002biochemistry}, signalling~\cite{hunter1995protein,tsai2009protein}, and the central dogma of molecular biology~\cite{crick1970central}. Analysing the mechanism by which they function is fundamental to understanding the exquisite behaviour of natural networks, to engineering existing systems~\cite{buckhout2018restriction,erb2017synthetic}, and to developing synthetic analogs {\it de novo}~\cite{rothlisberger2008kemp,jiang2008novo,siegel2010computational}. 

\teo{A catalytic enzyme enhances the overall rate of a molecular process by participating in reaction intermediates. The enzyme, however, is recovered unscathed at the end of the process. This mechanism of action allows a single enzyme to turn over a large number of metabolites, but it has equally fundamental consequences for signalling and information processing systems.

Catalytic action allows an enzyme to modify its substrates in a {\em persistent} way, so that the products do not immediately convert back once they detach from the enzyme -- unlike in simpler mechanisms of direct allosteric action \cite{Ulrich2005}. The importance of this persistence is exemplified by kinase signalling networks, which are central to signal transduction and information processing in eukaryotes \cite{Manning2002,Herskowitz1995}. In these systems, kinase enzymes catalyse phosphorylation of specific amino acid residues within protein substrates. The phosphorylated products demonstrate a change in activity, becoming, for example, activated transcription factors or kinases that in turn activate further downstream species. Crucially, the activated products do not need to remain bound to their upstream enzymes to be functional - their activation persists beyond the timescale of the catalytic interaction.

In this way, catalytic mechanisms allow the formation of (meta)stable memories in the states of downstream molecules, which retain information on their prior interactions \cite{Ouldridge_comput_2017}. These memories can suppress sensing noise through time-integration of signals \cite{govern2014}; underlie the calculation of time derivatives of an input in the context of adaptive chemotaxis \cite{Armitage1999}; and are central to the measurement and feedback cycles of recently-proposed molecular integral feedback controllers \cite{Briat2016,Encisco2019,Cappelletti2019}}.
Furthermore, since a single catalyst can activate multiple downstream substrates, catalyst motifs can be used to split and amplify signals~\cite{Mehta2016}.

Although catalysts are not consumed by reactions, they must transiently participate in the intermediate complexes. In so doing they face a central paradox: catalysts must bind strongly enough to participate in the reaction, but weakly enough to be recovered at the end. Moreover, whilst participating in a reaction, enzymatic catalysts are generally sequestered by their binding partners and cannot act on additional substrates. In some cases this sequestration may allow novel, advantageous behaviour. For example, sequestration is necessary for the mechanism of ``zero-order ultrasensitivity" that allows for sharp responses of the output substrate to small changes in concentration of the input catalyst~\cite{goldbeter1981amplified,huang1996ultrasensitivity,ferrell2014ultrasensitivity}. Frequently, however, sequestration is a potential disadvantage~\cite{jayanthi2011retroactivity}: it can limit the maximal rate of substrate turnover, and can cause ``retroactive" loading effects in signalling networks that lead to breakdowns in assumptions of modularity~\cite{tuanase2006signal,ventura2008hidden,del2008modular,barton2013energy,del2015biomolecular,deshpande2017high}. The latter is a particular concern in the rational design of synthetic systems~\cite{meijer2017hierarchical}.

The simple Michaelis-Menten model of enzymatic kinetics illustrates the problem of sequestration for substrate turnover~\cite{henri2006theorie,menten1913kinetik}. Consider the following reaction scheme:
\begin{eqnarray}
{\rm E} + {\rm S} \xrightleftharpoons[k_{-0}]{k_{+0}} {\rm ES} \xrightarrow[]{k_{+ \rm cat}} {\rm E} + {\rm P}.
\label{eq:MM}
\end{eqnarray}
Here E is the catalytic enzyme, S is the substrate, ES is the enzyme-substrate complex, and P is the product. The bimolecular rate constant $k_{+0}$ describes the speed with which substrates and enzymes bind; the unimolecular rate constants $k_{-0}$ and $k_{+ \rm cat}$ describe the rates of unbinding and product generation and release, respectively. Assuming the concentration of the ES complex, $[ES]$, reaches a quasi-steady state \teo{rapidly relative to any depletion of the substrate}~\cite{briggs1925note,schnell2014validity}, one obtains $k_{+0}[E][S] =(k_{-0}+k_{\rm +cat})[ES]$. The conservation law for the total enzyme concentration gives $[E]+[ES]=[E_{\rm tot}]$, implying that the concentration of enzymes in complexes with the substrate  is $[ES]=\frac{k_{+0}[E_{\rm tot}][S]}{k_{+0}[S]+ k_{-0}+k_{\rm +cat}}$. Therefore, the overall flux of substrates through the reaction is $r=k_{\rm +cat}[ES]=k_{\rm +cat}\frac{k_{+0}[E_{\rm tot}][S]}{k_{+0}[S]+k_{-0}+k_{\rm +cat}}=r_{\rm max}\frac{[S]}{[S] + K_M}$, where $r_{\rm max}=k_{\rm +cat}[E_{\rm tot}]$ is the {limiting} possible rate of the reaction and $K_M = k_{-0}/(k_{+0} +k_{\rm +cat})$. The quantity $\frac{[S]}{[S] + K_M} \leq 1$ reflects the effect of sequestration {in quasi-steady state}. {It is also known as the efficiency of the enzyme catalyzed reaction~\cite{schnell1997enzymological,schnell1997theoretical}}. At low substrate concentration $[S]\rightarrow 0$, the reaction flux is proportional to $[S]$; plenty of enzymes are available to process additional substrates at the same rate per substrate. For $[S] \gtrsim K_M$, however, the reaction flux plateaus because a substantial fraction of the enzymes become sequestered, and fewer are available to process additional substrates. 

\teo{In this article, we analyse how properties of an enzyme might be tuned, either by evolution or bioengineers, to achieve the goal of minimising sequestration while maintaining a certain reaction flux. As shown by the pioneering work of Terrell Hill~\cite{hill1966studies,hill1983some}, and emphasised in the more recent works of Beard, Qian and coworkers \cite{qian2003stoichiometric,beard2004thermodynamic,qian2005thermodynamics}, catalysts operate out of equilibrium and  understanding the fundamental constraints on their behaviour requires a thermodynamic perspective. Numerous authors have considered the implications of thermodynamics for catalytic function in various contexts, from a generic relationship between free energy and one-way reaction fluxes~\cite{beard2007relationship} to inherent trade-offs between thermodynamic cost and performance in sensing and signalling~\cite{govern2014,Mehta2016} and insulating motifs that are designed to suppress retroactive effects~\cite{deshpande2017high,barton2013energy}. 
 In our case, we will assume a fixed overall thermodynamic drive to a catalytic reaction, and instead explore how the thermodynamic stability of reaction intermediates determine whether a catalyst can achieve high substrate turnover rates at low levels of sequestration.}

The manuscript is structured as follows: in section~\ref{sec:modelling}, we justify the use of a specific thermodynamically self-consistent model of enzyme-substrate catalysis~\cite{chaplin1990enzyme,beard2008chemical,voet2011biochemistry}.
We then formulate our question as an optimization problem. We ask: how do we choose binding free energies  of the intermediate enzyme-substrate and enzyme-product complexes so that the system minimises the number of enzyme-substrate complexes whilst maintaining a required minimum flux of reactants into products? In Section~\ref{sec:results}, we analyse this optimization problem under the assumption that association reactions are diffusion limited, finding that there is an inherent trade-off in such motifs. Choosing very high binding free energies for the intermediate complexes reduces retroactivity, but also reduces the flux through the circuit. On the other hand, choosing very low binding energies implies that the system spends a large proportion of time in the intermediate states, increasing the retroactivity of the system. We show that the optimal binding free energies are not only moderate as a consequence of this trade-off, but they are strongly related to each other. In particular, the difference between the optimal binding free energies is a constant that is related to the intrinsic free energy difference between the products and reactants. In addition, we also show that the optimal circuit saturates the bound on the flux requirement. In Section~\ref{sec:non_diffusion}, we relax the assumption that binding rates are diffusion limited. We find that many of our observations from the diffusion-limited regime carry over qualitatively to this new regime. The optimal binding free energies are still moderate and the difference between them is confined to a value close to the intrinsic free energy difference between the products and reactants.

\section{Model and methods}\label{sec:modelling}

\begin{figure}[htbp]
\centering
\includegraphics[scale=0.25]{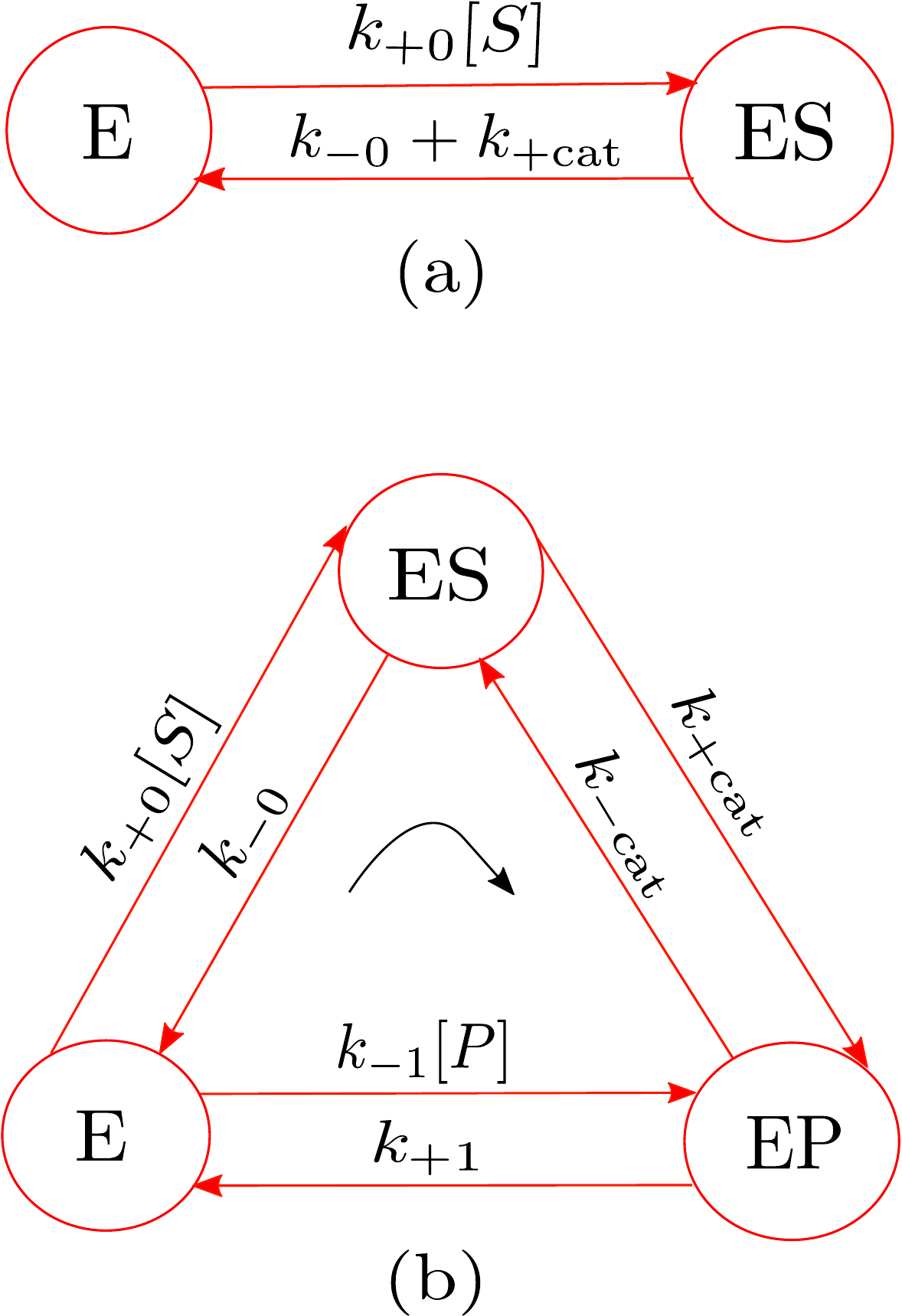}
\caption{Markov chains for the evolution of an isolated enzyme's binding state in minimal models of catalysis. (a) Markov model  for the catalytic mechanism ${\rm E} + {\rm S} \xrightleftharpoons[k_{-0}]{k_{+0}} {\rm ES} \xrightarrow[]{k_{+ \rm cat}} {\rm E} + {\rm P}$. The single $\rm ES \rightarrow E$ transition in the stochastic process includes two physically distinct processes; $\rm ES \rightarrow E + S$ and $\rm ES \rightarrow E+P$. (b) Markov chain corresponding to the enzyme-substrate catalysis given by Equation~\ref{eq:actual_mechanism}: ${\rm E}+{\rm S}\xrightleftharpoons[k_{-0}]{k_{+0}}{\rm ES}\xrightleftharpoons[k_{-\rm{cat}}]{k_{+\rm{cat}}}{\rm EP}\xrightleftharpoons[k_{-1}]{k_{+1}} {\rm E}+ {\rm P}$. The completion of single clockwise cycle converts a substrate molecule into a product.}
\label{fig:schematic}
\end{figure}

We now introduce basic modelling assumptions, in the process explaining why the classic Michaelis-Menten model is insufficient for our purposes. Henceforth, we will use natural units in which $k_{\rm B}T=1$, all rates are defined dimensionlessly relative to $1$\,s$^{-1}$, and all concentrations given dimensionlessly relative to 1M. We model dilute biochemical systems at the level of molecular macrostates~\cite{ouldridge2018importance}; reactions are described by mass-action kinetics with well-defined rate constants. We consider an ensemble of enzymes interacting with substrate S and product P molecules; \teo{these substrates and products are assumed to have approximately constant concentrations $[S]$ and $[P]$ on the timescale of interest}~\cite{ge2013dissipation}. In this limit, the trajectory of a single enzyme through its discrete binding states can be analysed independently as a continuous time Markov chain, with pseudo-first-order transition rates that depend on $[S]$ and $[P]$~\cite{wachtel2018thermodynamically}. The resultant probabilities are proportional to the expected concentrations of enzymatic states in a bulk system. The systems we consider form irreducible Markov chains, and therefore tend toward a well-defined steady-state probability distribution $\pi_i$ describing the occupancy of enzyme binding states $i$. 

We will first illustrate our approach with the commonly-used model of Eq.~\ref{eq:MM}, before arguing that it is insufficiently rigorous to allow a meaningful optimisation. We will then present the extended model that will form the basis of this work. The model of Eq.~\ref{eq:MM} has a Markov chain representation shown in Fig.~\ref{fig:schematic}\,(a). The binding states of the enzyme are unbound (E) and substrate-bound (ES). Both the release of the product and unbinding of the substrate contribute to the same transition (ES to E) at the level of the enzyme's binding states. The probability of the enzyme being unbound (equal to the fraction of unbound enzymes in an ensemble) is $\pi_{E}$, and the net rate of product output per enzyme is $k_{+ \rm cat} \pi_{ ES}$. 

We consider the challenge of optimizing the enzyme properties to achieve a desired steady state rate of conversion of S into P  per enzyme, $\Psi$, with a minimal steady-state fraction of sequestered enzymes, $\mathcal{R}$, at fixed concentrations $[S]$ and $[P]$. For the model in Fig.~\ref{fig:schematic}\,(a), these quantities are given by $\Psi=k_{+ \rm cat} \pi_{ ES}$ and  $\mathcal{R}=1 - \pi_{E}$, respectively. Furthermore, since $\pi_{ E} = \frac{k_{\rm +cat}+k_{-0}}{k_{+0}[S]+k_{\rm +cat}+k_{-0}}$ and $\pi_{ ES}=\frac{k_{+0}[S]}{k_{+0}[S]+k_{\rm +cat}+k_{-0}}$ , it is immediately clear that the sequestration fraction $\mathcal{R}$ can be made arbitrarily small, without compromising the flux $\Psi$, by allowing the catalytic rate $k_{\rm +cat}\rightarrow\infty$. To obtain meaningful insight, it is therefore necessary to consider physically-motivated constraints on the kinetic parameters. 

The most important constraint is that since the enzyme E is not consumed in the reaction, its properties cannot influence the overall free-energy change of reaction, $\Delta G = \ln \frac{[P]}{[S]} - \Delta \mu$. Here, $\Delta \mu$ is the intrinsic free-energy difference between S and P; by our sign convention, a positive $\Delta \mu$ implies P is more thermodynamically stable than S, providing a forward drive to the reaction. Note that $\Delta \mu$ can also incorporate the contribution to $\Delta G$ of the consumption of ancillary fuel molecules, such as ATP, which are treated implicitly in the model of Eq.~\ref{eq:MM}. If the environment of substrates and products is fixed, an optimization over enzyme properties corresponds to optimizing at fixed $\Delta G$. 

For a single reaction step $j$, the principle of detailed balance dictates that the free energy change is directly related to the forwards and backwards transition rates $\nu_{\pm j}$~\cite{ouldridge2018importance,frenkel2001understanding}:
\begin{eqnarray}
\Delta G_j = - \ln  \frac{ \nu_j}{\nu_{-j}}.
\label{eq:db}
\end{eqnarray}
For a multi-step reaction, one can simply add together Eq.~\ref{eq:db} for each step $j$, obtaining
\begin{eqnarray}
\Delta G = - \ln \prod_j \frac{ \nu_j}{\nu_{-j}}.
\end{eqnarray}
 Since the catalytic step of the Michaelis-Menten model has no reverse complement, $\Delta G$ is undefined, making it impossible to perform an optimization at  fixed $\Delta G$. It is therefore necessary to introduce a backwards transition, which would allow E and P to bind, and be converted into ES~\cite{fisher1999force,astumian2015irrelevance,haldane1930course,sauro2011enzyme}. 

Having included this reaction, it is hard to justify combining both the chemical conversion of substrate into product, and its release from the enzyme, in a single step, as in Eq.~\ref{eq:MM}. If both P and S can be converted into each other by E, shouldn't the binding and unbinding of P also be treated explicitly using a binding state EP? Indeed, ignoring EP corresponds to assuming that the enzyme-product complex is arbitrarily short-lived, yet does not present a barrier to the conversion of S into P. This assumption seems to ignore the very challenge of the optimization problem itself. {We therefore use the following model~\cite{chaplin1990enzyme,beard2008chemical,voet2011biochemistry}},
\begin{align}\label{eq:actual_mechanism}
{\rm E}+{\rm S}\xrightleftharpoons[k_{-0}]{k_{+0}}{\rm ES}\xrightleftharpoons[k_{-\rm{cat}}]{k_{+\rm{cat}}}{\rm EP}\xrightleftharpoons[k_{-1}]{k_{+1}} {\rm E}+ {\rm P},
\end{align}
as the minimal description of catalysis in which we can meaningfully ask how enzyme properties can be adjusted to minimize sequestration $\mathcal{R}$ at fixed flux per enzyme ${\Psi}$. This molecular model can be represented as a continuous time Markov process over the enzymatic states E, ES and EP as shown in Fig.~\ref{fig:schematic}\,(b); within this description, we obtain
\begin{align}
\mathcal{R} = 1 - \pi_{E} = \pi_{ ES} + \pi_{EP},
\label{eq:retroactivity}
\end{align}
and
\begin{align}
{\Psi} =  k_{+0}[S] \pi_{ E} - k_{-0} \pi_{ES}.
\end{align}
The requirement of fixed $\Delta G$ equates to 
\begin{align}
\frac{k_{+0}[S] k_{\rm + cat} k_{+1}}{k_{-0} k_{\rm - cat} k_{-1}[P]} = \exp(-\Delta G) = {\rm const.}
\label{eq:constraint}
\end{align}

Even with the restriction to fixed $\Delta G$, and fixed concentrations $[S]$ and $[P]$, the optimization problem is still poorly constrained. As it stands, all rate constants in Eq.~\ref{eq:constraint} could be increased by an arbitrary factor, allowing an arbitrarily high $\Psi$ whilst leaving the stationary distribution $\pi$ (and hence $\mathcal{R}$) unchanged. It would therefore be possible to obtain any flux $\Psi$ whilst ensuring $\mathcal{R} \rightarrow 0$. In practice, \teo{molecular processes cannot be arbitrarily fast due to basic constraints from physical chemistry. In any case, it is far easier to tune, whether by rational engineering or evolution, the binding free energies of the metastable substrate-bound and product-bound states than it is to further optimise the precise chemistry of the ensemble of transition states between the two.} To reflect this fact, we restrict ourselves to optimizing the two standard binding free energies $\Delta G_{ES}$ and $\Delta G_{EP}$, which are related to the rate constants by 
\begin{align}
\exp(-\Delta G_{ES}/k_{\rm B}T)= \frac{k_{+0}}{k_{-0}}, & \nonumber \\
\exp(-\Delta G_{EP}/k_{\rm B}T)= \frac{k_{-1}}{k_{+1}}, &  \nonumber \\
\exp \left((-\Delta G_{EP}+\Delta G_{ES}+ \Delta \mu)/k_{\rm B}T \right) = \frac{k_{\rm + cat}}{k_{\rm -cat}}. \label{eq:rates}&
\end{align}
A schematic representation of the the free energy profile of this process, and its dependence on the parameters in Eq.~\ref{eq:rates}, is given in Fig.~\ref{fig:reaction_coordinate}.

\begin{figure}[htbp]
\centering
\includegraphics[scale=0.16]{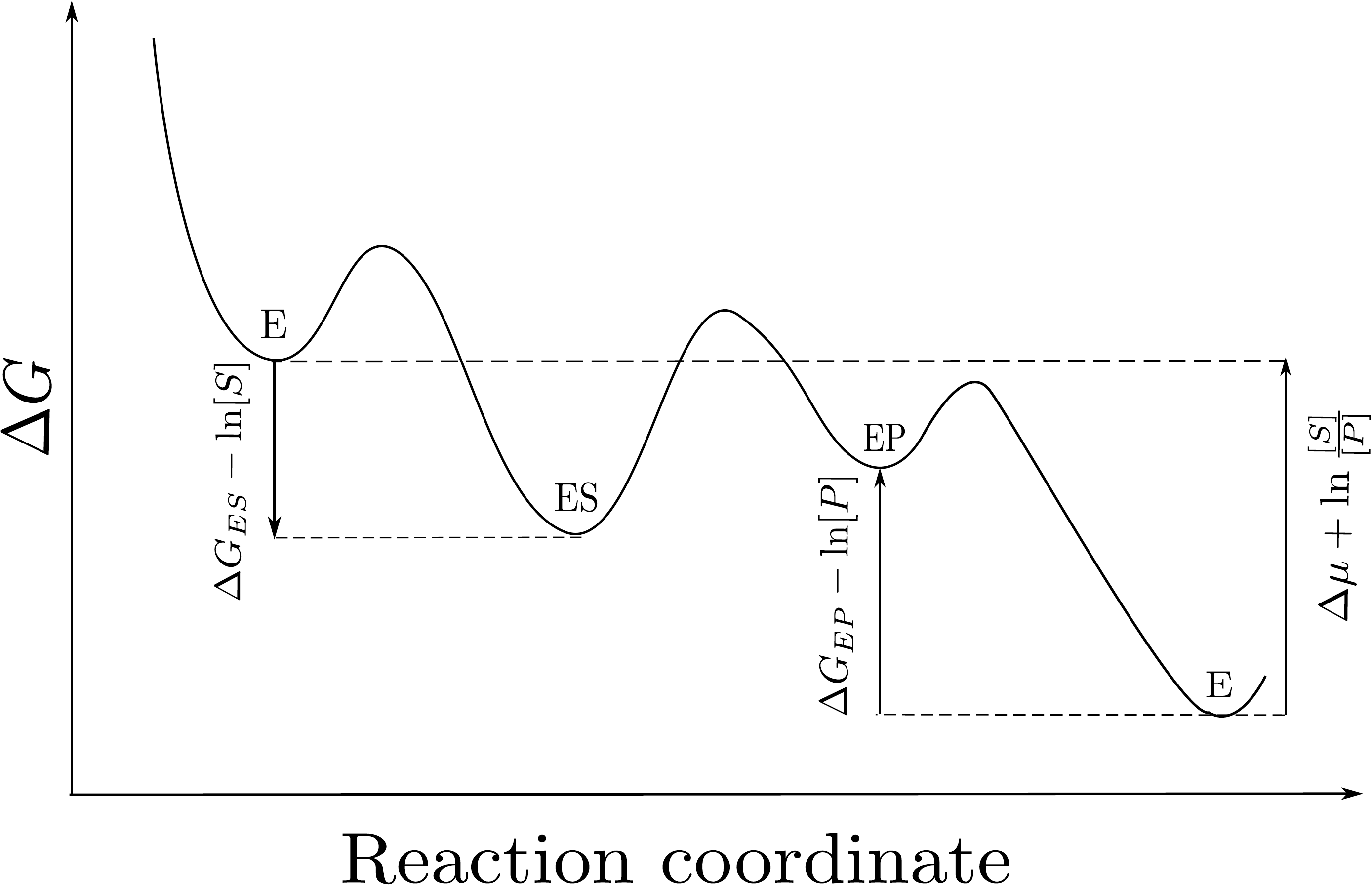}
\captionof{figure}{Illustrative free-energy profile of the catalytic conversion of S into P by E, with the number of S and P molecules present explicitly accounted for. Chemical macrostates are shown as local minima in the profile, separated by barriers. A full catalytic cycle corresponds to moving from the leftmost E minimum to the rightmost, consuming S and producing P. Although the enzyme itself returns to its original state, the process as a whole is thermodynamically downhill ($\Delta G = \ln[P]/[S] - \Delta \mu$) because the final state has one fewer substrate and one more product molecule. Optimization corresponds to adjusting the heights of the metastable ES and EP bound states to minimise sequestration whilst maintaining a fixed flux.}
\label{fig:reaction_coordinate}
\end{figure}

A given set of binding free energies, along with $\Delta \mu$, therefore specifies ratios of forwards/backwards rate constants. Absolute values remain ambiguous. To make progress, we must make further assumptions about how rate constants respond to changes in the binding free energies. A similar issue was addressed in Ref.~\cite{brown2017allocating}, in which the authors attempted to optimize the flux of trajectories through a series of states by adjusting their energies (in that work, unlike this one, there was no attempt to minimise occupancy of intermediate states).  In that case, one of the two transitions in a backwards/forwards pair was assumed to be exponentially sensitive (``labile") to the energy difference, whereas the other was assumed to be constant. 

In our system, we will initially assume (for simplicity) that the binding reactions are diffusion-controlled~\cite{atkins2011physical,alberty1958application}; i.e. the on-rate is fixed by the diffusion time scales that are independent of the details of the enzyme's interaction wih substrate and product. In particular, we set $k_{+0}=k_0 = {\rm const}$ and $k_{-1}=k_1 = {\rm const}$. In the language of~\cite{brown2017allocating}, the substrate binding transition is ``backwards labile" and the product release transition is ``forwards labile". There is no immediately obvious reason to make the intermediate step of chemical catalysis either forward or backward labile. We show in Section~\ref{sec:results}, however, that both choices give pathological results for the question we ask. Invoking the fact that a true chemical modification cannot happen arbitrarily fast, we then consider an alternative in which both the forward and backward catalytic rate constants have a finite upper bound:
\begin{align}
k_{+\rm{cat}}=k_{\text{cat}}\min(1,e^{- \Delta G_{EP}+\Delta G_{ES} + \Delta\mu + \Delta G_c}),
\label{eq:k+cat}
\end{align}
and
\begin{align}
k_{-\rm{cat}}=k_{\text{cat}}e^{\Delta G_c}\min(1,e^{-\Delta G_{ES}+ \Delta G_{EP} - \Delta\mu -\Delta G_c}),
\label{eq:k-cat}
\end{align}
for some $\Delta G_C\in\mathbb{R}$. The scheme is based on ``Metropolis dynamics", in which reactions that are downhill in free energy have a fixed rate and uphill reactions are slowed down \cite{srinivas2013biophysics}. This dynamics is equivalent to assuming that, in an Arrhenius picture, the transition state is a fixed free energy above the maximum free energy of the two metastable states on either side. The inclusion of a finite $\Delta G_C$ generalises this approach to allow an offset between the maximum forward and backward catalytic rate constants. The overall effect is to split the $\Delta G_{ES}-\Delta G_{EP}$ plane into two regions: region $I$ in which the interconversion is backwards labile and the backwards step is slow; and region $II$ in which the interconversion is forwards labile while the forwards reaction is slow. A graphical illustration of these different responses to the free energies of transition is given in Fig.~\ref{fig:rate_constant_energy}. We are now able to fully state our main optimization problem:

\begin{question}\label{pr:original}
\textit{How should binding free energies $\Delta G_{ES}$ and $\Delta G_{EP}$ in the model of Eq.~\ref{eq:actual_mechanism} be chosen to minimize sequestration $\mathcal{R}$, whilst maintaining a product output rate of $\Psi_0$, given fixed $[S]$, $[P]$ and $\Delta \mu$, diffusion-limited binding reactions, and catalytic steps with fixed and finite upper bounds on their rates (Eqs.~\ref{eq:k+cat} and \ref{eq:k-cat})?} 
\end{question}

\begin{figure}[htbp]
\centering
\includegraphics[scale=0.37]{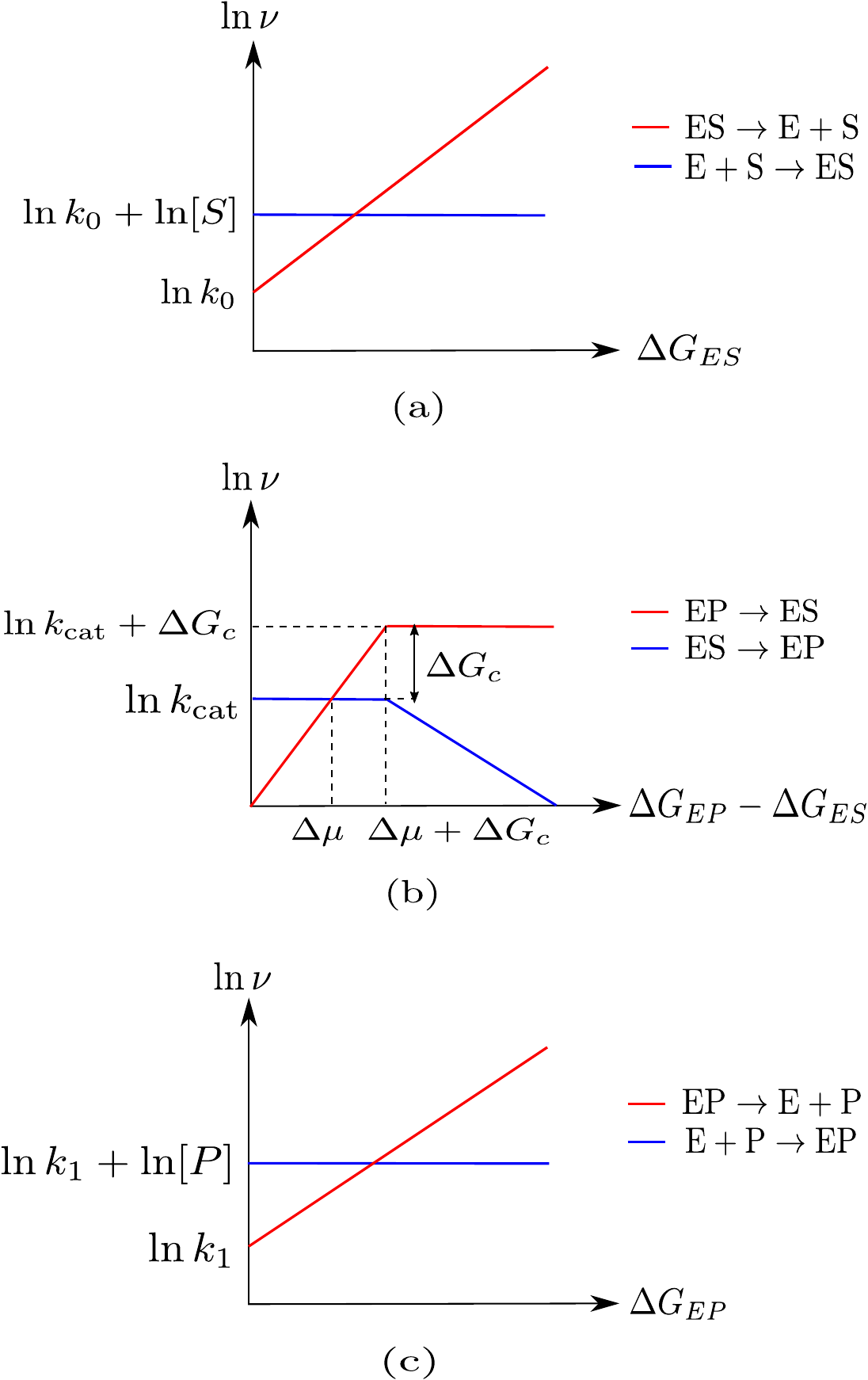}
\captionof{figure}{Graphical illustration of the dependence of the transition rates $\nu$ of the Markov chain in Figure~\ref{fig:schematic}\,(b) on binding free energies. Binding rate constants are assumed to be fixed due to diffusion. Catalytic rate constants are bounded, whilst maintaining $\frac{k_{+\rm{cat}}}{k_{-\rm{cat}}}=e^{-\Delta G_{EP} + \Delta G_{ES} + \Delta\mu}$. In this case, the $\Delta G_{c}$ parameter fixes an offset between the maximal rates.}\label{fig:rate_constant_energy}
\end{figure}

\section{Results and Discussion}\label{sec:results}

\subsection{Diffusion-controlled binding rates}
\label{sec:diff-controlled}
To answer Problem~\ref{pr:original}, we seek the optimal binding free energies  $\Delta G^{\rm opt}_{ES}$ and $\Delta G^{\rm opt}_{EP}$ that achieve a minimal sequestration $\mathcal{R}^{\rm opt}$, whilst maintaining an output flux $\Psi^{\rm opt} \geq \Psi_0$. Here $\Psi_0\in\mathbb{R}_{>0}$ is the target seady-state net rate of substrate turnover. We will first present a mathematical analysis, followed by a physical explanation and interpretation.  

\subsubsection{Detailed analysis of diffusion-controlled system}
\label{sec:analysis1}

We describe a target flux $\Psi_0$ as {\em achievable} if choices of $\Delta G_{ES}$ and $\Delta G_{EP}$ exist that satisfy $\Psi \geq \Psi_0$.   We will show that solutions to Problem~\ref{pr:original} for achievable target fluxes $\Psi_0$  lie on the line $\Delta G^{\rm opt}_{EP} = \Delta G^{\rm opt}_{ES}+ \Delta\mu + \Delta G_c$, with $\Delta G^{\rm opt}_{ES}$ and $\Delta G^{\rm opt}_{EP}$ taken as high as possible to saturate the flux constraint, $\Psi^{\rm opt}=\Psi_0$. To make these arguments, we require machinery from the theory of Markov chains. To guide this derivation, we illustrate the continuous time Markov chain corresponding to Eq.~\ref{eq:actual_mechanism}, with its explicit dependence on the parameters $\Delta G_{ES}$, $\Delta G_{EP}$, $\Delta \mu$ and $\Delta G_c$ as laid out in Eq.~\ref{eq:rates}, in Fig.~\ref{fig:diffusion}. We first solve for $\mathcal{R}$ and $\Psi$ in terms of basic properties of this Markov chain.

\begin{figure}[htbp]
\centering
\includegraphics[scale=0.4]{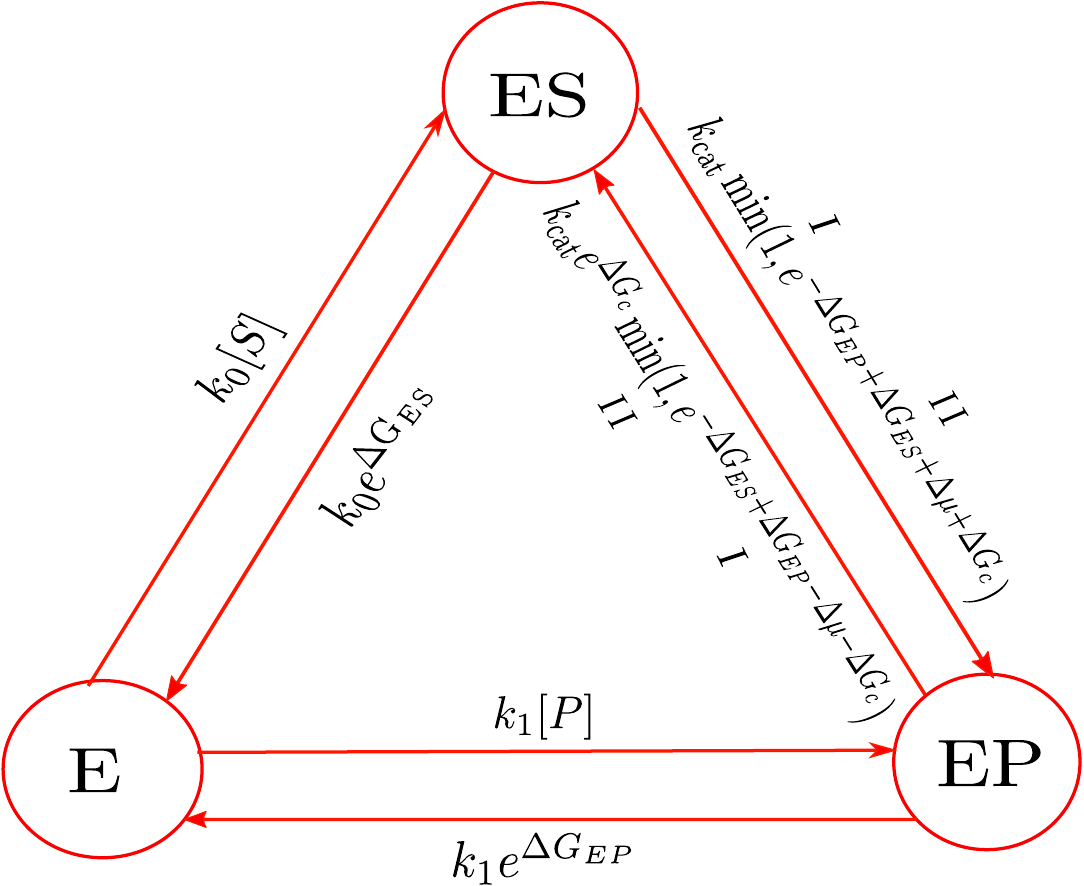}
\captionof{figure}{Markov chain for Problem~\ref{pr:original}, showing the explicit dependence of rate and the optimisation variables 
$\Delta G_{ES}$ and $\Delta G_{EP}$, along with other parameters that are constant during the optimisation of a given system. The options inside the minimization statement of the catalytic rate constants have been labelled corresponding to appropriate regions ($I$ or $II$) in the $\Delta G_{ES}-\Delta G_{EP}$ plane.}\label{fig:diffusion}
\end{figure}

Let us denote the expected lifetime of state $i$ by $\tau_i$; and the expected time for the next arrival at $j$ given that the current state is $i$ (the mean first passage time) by $\tau_{i \rightarrow j}$. The series of states visited by the enzyme form an embedded discrete-state, discrete-time Markov chain~\cite{norris1998markov}; let $P_{i \rightarrow j}$ represent the transition probabilities of this Markov chain. $P_{i \rightarrow j}$ is then the probability that $j$ is the next state visited given that the system is in state $i$. By the memoryless property of continuous time Markov chains, the lifetime of a state is exponentially distributed with parameter equal to the total rate of outward transition from that state, implying that $\tau_i=\frac{1}{\sum_j K_{ji}}$, where $K_{ji}$ is the rate of outward transition from state $i$ to $j$.

Since we have an irreducible Markov chain, using \cite[Theorem~3.8.1]{norris1998markov} and Equation~\ref{eq:retroactivity}, we obtain 
\begin{align}
\mathcal{R} = 1- \pi_{E} = 1 - \frac{\tau_{ E}}{\tau_{E\rightarrow E}}.
\end{align}
Using $\tau_{E \rightarrow E} = \tau_E + P_{ E \rightarrow ES} \tau_{ES\rightarrow E} + P_{E\rightarrow EP} \tau_{EP \rightarrow E}$, and noting that $P_{E\rightarrow ES}=k_{0}[S]\tau_E$ and $P_{E\rightarrow EP}=k_{1}[P]\tau_E$, we find
\begin{eqnarray}\label{eq:retroactivity_intermitent}
\mathcal{R} = 1 - \frac{1}{1 + k_{0}[S]\tau_{ES\rightarrow E} + k_{1}[P]\tau_{EP \rightarrow E}}.
\end{eqnarray}

Observing that 
\begin{align}
\tau_{ES\rightarrow E}=  \tau_{ES} + P_{ES \rightarrow EP} \tau_{EP\rightarrow E}, &\nonumber \\
\tau_{EP\rightarrow E}=  \tau_{EP} + P_{EP \rightarrow ES} \tau_{ES\rightarrow E}, &
\end{align}
 we can solve for the average first passage times as
\begin{align}
\tau_{ES\rightarrow E}=  \frac{\tau_{ES} + P_{ES \rightarrow EP} \tau_{EP}}{1-P_{EP \rightarrow ES} P_{ES \rightarrow EP} }, &\nonumber \\
\tau_{EP\rightarrow E}=  \frac{\tau_{EP} + P_{EP \rightarrow ES} \tau_{ES}}{1-P_{EP \rightarrow ES} P_{ES \rightarrow EP} }, &
\end{align}
As a result, we can re-write the sequestered fraction in Eq.~\ref{eq:retroactivity_intermitent}  solely in terms of properties of single transition steps:
\begin{align}\label{eq:retroactivity_probability}
& \mathcal{R} = 1 - \nonumber \\
& \frac{1}{1 + k_{0}[S]\Big(\frac{\tau_{ES} + P_{ES \rightarrow EP} \tau_{EP}}{1-P_{EP \rightarrow ES} P_{ES \rightarrow EP}}\Big) + k_{1}[P]\Big(\frac{\tau_{EP} + P_{EP \rightarrow ES} \tau_{ES}}{1-P_{EP \rightarrow ES} P_{ES \rightarrow EP} }\Big)}.
\end{align}

For the flux $\Psi$, there are two cases:
\begin{enumerate}
\item Region $I: \Delta G_{EP} < \Delta G_{ES}+  \Delta\mu +\Delta G_c$, ${\rm ES} \rightleftharpoons {\rm EP} $ backward labile,
\begin{align}
\Psi_b = \frac{(1-\mathcal{R})k_0k_1k_{\text{cat}}[E_{\text{tot}}]([S]e^{\Delta\mu + \Delta G_c} - [P])}{k_0k_1e^{\Delta\mu + \Delta G_c + \Delta G_{ES}} + k_0k_{\text{cat}} + k_1k_{\rm cat}e^{\Delta\mu + \Delta G_c}}
\label{eq:flux_case_1}
\end{align}
\item Region $II: \Delta G_{EP} > \Delta G_{ES}+  \Delta\mu +\Delta G_c$, ${\rm ES} \rightleftharpoons {\rm EP} $ forward labile, 
\begin{align}
\Psi_f = \frac{(1-\mathcal{R})k_0k_1k_{\text{cat}}[E_{\text{tot}}]([S]e^{\Delta\mu + \Delta G_c} - [P])}{k_0k_1e^{\Delta G_{EP}} + k_0k_{\text{cat}} + k_1k_{\rm cat}e^{\Delta\mu + \Delta G_c}}
\label{eq:flux_case_2}
\end{align}
\end{enumerate}
The division of $ \Delta G_{ES}- \Delta G_{EP}$ space into these two regions is shown schematically in Fig.~\ref{fig:diffusion_energy_regions}. For the edge case of $\Delta G_{EP} = \Delta G_{ES}+  \Delta\mu +\Delta G_c$ that divides the two regions, both $\Psi_b = \Psi_f$ are valid. We now state a few lemmas that describe the behaviour of retroactivity and flux with respect to changing binding energies in regions $I$ and $II$.

\begin{lemma}\label{lem:region_I_behaviour}
In Region $I$, increasing $\Delta G_{EP}$ decreases $\mathcal{R}$ and increases $\Psi$.
\end{lemma}

\begin{proof}
The transition probabilities $P_{ES \rightarrow EP}$ and $P_{EP \rightarrow ES}$ are both independent of $\Delta G_{EP}$ in region $I$, since $\Delta G_{EP}$ is irrelevant to the transitions out of state ES, and contributes the same factor to all transitions out of EP:
\begin{align}
P_{ES \rightarrow EP} = \frac{k_{\rm cat}}{k_0e^{\Delta G_{ES}}+k_{\rm cat}},&\nonumber \\
P_{EP \rightarrow ES} =   \frac{k_{\text{cat}}e^{-\Delta G_{ES}  -\Delta\mu }}{k_1 + k_{\text{cat}}e^{-\Delta G_{ES}  -\Delta\mu }}.&
\end{align}
The expected lifetime of EP, $\tau_{EP}=\frac{1}{{\rm e}^{\Delta G_{EP}}(k_1 + k_{\text{cat}}e^{-\Delta G_{ES} -\Delta\mu})}$, decreases monotonically with increasing ${\Delta G_{EP}}$, while $\tau_E=\frac{1}{k_0[S]+k_1[P]}$ and $\tau_{ES}=\frac{1}{k_{\text{cat}} + k_0e^{\Delta G_{ES}}}$ remain unchanged. Eq.~\ref{eq:retroactivity_probability} therefore shows that $\mathcal{R}$ must decrease as $\Delta G_{EP}$ increases within region $I$, and Eq.~\ref{eq:flux_case_1}, shows flux $\Psi$ increases with $\Delta G_{EP}$ in region $I$. 
\end{proof}

\begin{lemma}\label{lem:region_II_behaviour}
In Region $II$, increasing $\Delta G_{ES}$ decreases $\mathcal{R}$ and increases $\Psi$.
\end{lemma}

\begin{proof}
The transition probabilities $P_{ES \rightarrow EP}$ and $P_{EP \rightarrow ES}$ are both independent of $\Delta G_{ES}$ in region $II$, since $\Delta G_{ES}$ is irrelevant to the transitions out of state EP, and contributes the same factor to all transitions out of ES:
\begin{align}
P_{EP \rightarrow ES} = \frac{k_{\rm cat}e^{\Delta G_c}}{k_1e^{\Delta G_{EP}}+k_{\rm cat}e^{\Delta G_c}},&\nonumber \\
P_{ES \rightarrow EP} =   \frac{k_{\text{cat}}e^{-\Delta G_{EP} + \Delta\mu + \Delta G_c}}{k_0  + k_{\text{cat}}e^{-\Delta G_{EP}+ \Delta\mu + \Delta G_c}}
\end{align}
The expected lifetime of ES, $\tau_{ES}=\frac{1}{{\rm e}^{\Delta G_{ES}}(k_0 + k_{\text{cat}}e^{-\Delta G_{EP} +\Delta\mu +\Delta G_c})}$, decreases monotonically with increasing ${\Delta G_{ES}}$, while $\tau_E=\frac{1}{k_0[S]+k_1[P]}$ and $\tau_{EP}=\frac{1}{k_{\text{cat}}e^{\Delta G_c} + k_1e^{\Delta G_{EP}}}$ remain unchanged. Eq.~\ref{eq:retroactivity_probability} therefore shows that $\mathcal{R}$ must decrease as $\Delta G_{ES}$ increases within region $I$, and Eq.~\ref{eq:flux_case_2}, shows flux $\Psi$ increases with $\Delta G_{ES}$ in region $II$.
\end{proof}

\begin{theorem}\label{thm:optimal_binding_saturation}
$\Delta G^{\rm opt}_{EP}=\Delta G^{\rm opt}_{ES} + \Delta\mu + \Delta G_c$.
\end{theorem}

\begin{proof}
For contradiction, assume not. Then, the optimal solution either lies inside region $I$ or region $II$.

\begin{figure*}[htbp]
\centering
\includegraphics[scale=0.08]{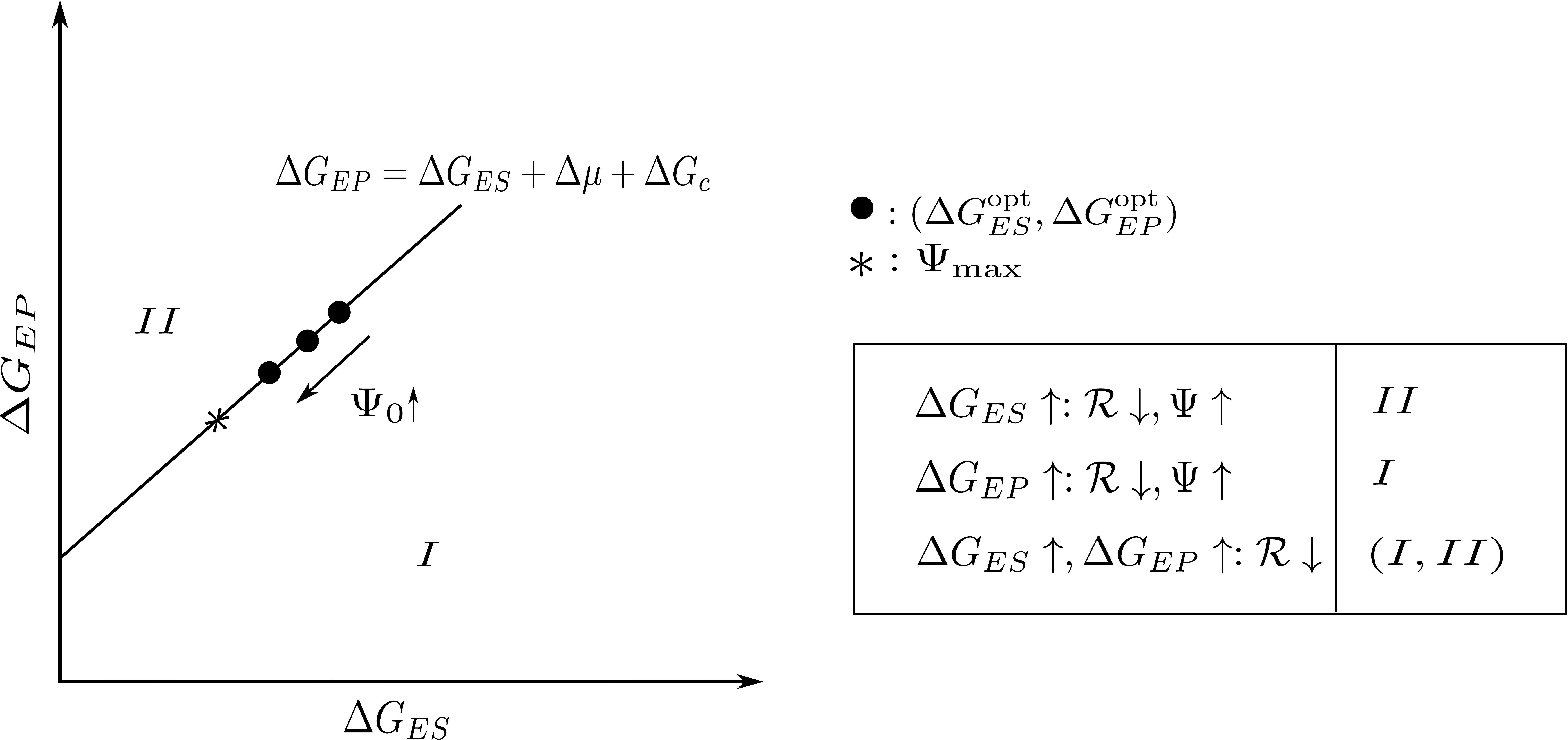}
\captionof{figure}{Graphical illustration of the typical behaviour of the optimisation problem~\ref{pr:original} within the $\Delta G_{ES}$-$\Delta G_{EP}$ plane. Within regions $I$ and $II$, it is always favourable to move towards the separating line $\Delta G^{\rm opt}_{EP} = \Delta G^{\rm opt}_{ES}+ \Delta\mu + \Delta G_c$, as doing so both increases flux $\Psi$ and reduces sequestration $\mathcal{R}$. Throughout the plane, including on the line, increasing both $\Delta G_{ES}$ and $\Delta G_{EP}$ by the same amount reduces retroactivity. These relationships are described in the box to the right of the figure. Both very positive and very negative values of $\Delta G_{ES}$ and $\Delta G_{EP}$ suppress flux leading to a maximal flux $\Psi_{\rm max}$ at moderate values indicated by a $*$ on the plot. Optimal solutions $(G^{\rm opt}_{ES},G^{\rm opt}_{EP})$  are denoted by black dots and are found on the line $\Delta G^{\rm opt}_{EP} = \Delta G^{\rm opt}_{ES}+ \Delta\mu + \Delta G_c$, to the right of the maximal flux, trading off flux against sequestration. The optimal solutions tend towards the cross $*$ as the target flux $\Psi_0$ increases towards $\Psi_{\rm max}$.}
\label{fig:diffusion_energy_regions}
\end{figure*}

Consider region $I$: assume an optimal pair of binding free energies $\left(\Delta G^{\rm opt}_{ES},\Delta G^{\rm opt}_{EP}\right)$, $\Delta G^{\rm opt}_{EP} < \Delta G^{\rm opt}_{ES}+  \Delta\mu +\Delta G_c$, exists that gives a minimal $\mathcal{R} = \mathcal{R}^{\rm opt}$ while satisfying $\Psi \geq \Psi_0$. By Lemma~\ref{lem:region_I_behaviour}, it is always possible to add $\delta G>0$ to $\Delta G^{\rm opt}_{EP}$ and increase $\Psi$ while reducing $\mathcal{R}$. The resulting pair $\left(\Delta G^{\rm opt}_{ES},\Delta G^{\rm opt}_{EP} + \delta G\right)$ necessarily satisfies $\Psi > \Psi_0$, and has a sequestered fraction $\mathcal{R} < \mathcal{R}^{\rm opt}$. The optimality of   $\left(\Delta G^{\rm opt}_{ES},\Delta G^{\rm opt}_{EP}\right)$ is therefore contradicted. 

An exactly analogous argument can be made for optimal solutions in region $II$. Using Lemma~\ref{lem:region_II_behaviour}, it is always possible to increase $\Delta G^{\rm opt}_{ES}$ by $\delta G>0$ and simultaneously reduce the sequestration fraction $\mathcal{R}$ and increase the flux $\Psi$.  Therefore no pair of binding free energies $\left(\Delta G^{\rm opt}_{ES},\Delta G^{\rm opt}_{EP}\right)$ within region $II$ can be optimal. As a result, optimal solutions must lie on the line dividing the regions, $\Delta G^{\rm opt}_{EP}=\Delta G^{\rm opt}_{ES} + \Delta\mu + \Delta G_c$.

\end{proof}

We note in passing that had we modelled the catalytic reaction as uniformly forwards labile (or uniformly backwards labile) throughout the $\Delta G_{ES}-\Delta G_{EP}$ plane, then pathological results would have been obtained. This choice would correspond to setting the whole of the plane to be region $II$ (or region $I$). As a result, it would always be possible to improve the design by increasing $\Delta G_{ES}$ (or $\Delta G_{EP}$), leading to divergent solutions that require unphysical, infinitely-fast transitions between ES and EP.

We now prove in Theorem~\ref{thm:optimal_flux_saturation} that if the flux constraint is achievable by the system, it is saturated: $\Psi_{\rm opt}=\Psi_0$. To do so, we first prove the following lemma

\begin{lemma}\label{lem:retroactivity_decrease_binding_energies}
$\mathcal{R}$ decreases if both binding free energies $\left(\Delta G_{ES},\Delta G_{EP}\right)$ are increased by the same amount $\delta G>0$. 
\end{lemma}

\begin{proof}
First, note that if both binding free energies are increased by the same $\delta G$, the new binding energies are in the same region of the phase plane shown in Fig.~\ref{fig:diffusion_energy_regions} as the old binding energies. Transition rates between ES and EP do not change, since the free energy change between the two binding states is unchanged. Transitions from $EP$ and $ES$ to $E$ are accelerated by $\exp(\delta G)$. As a result of these changes in rates, $\tau_{ES}$, $\tau_{EP}$, $P_{ES \rightarrow EP}$ and $P_{EP \rightarrow ES}$ all necessarily decrease for $\delta G>0$. It immediately follows from Eq.~\ref{eq:retroactivity_probability} that $\mathcal{R}$ decreases if both binding energies $\left(\Delta G_{ES},\Delta G_{EP}\right)$ are increased by $\delta G>0$.
\end{proof}

Lemmas~\ref{lem:region_I_behaviour},~\ref{lem:region_II_behaviour},~\ref{lem:retroactivity_decrease_binding_energies} and Theorem~\ref{thm:optimal_binding_saturation} are summarised in Fig.~\ref{fig:diffusion_energy_regions}. 

\begin{theorem}\label{thm:optimal_flux_saturation}
$\Psi^{\rm opt}=\Psi_0$.
\end{theorem}
 
\begin{proof} 
For contradiction, assume an optimal pair $\left(\Delta G^{\rm opt}_{ES},\Delta G^{\rm opt}_{EP}\right)$ with $\Psi^{\rm opt} > \Psi_0$ exists. By Lemma~\ref{lem:retroactivity_decrease_binding_energies}, a pair with $\left(\Delta G^{\rm opt}_{ES}+\delta G,\Delta G^{\rm opt}_{EP} + \delta G\right)$ with lower sequestration $\mathcal{R} < R^{\rm opt}$ can be found for arbitrary $\delta G>0$. Since $\Psi$ is a continuous function, it is always possible to choose a sufficiently small $\delta G$ such that the new system has $\Psi \geq \Psi_0$, contradicting our initial assumption. Therefore an optimal pair with finite $\left(\Delta G^{\rm opt}_{ES},\Delta G^{\rm opt}_{EP}\right)$ and $\Psi^{\rm opt} > \Psi_0$ cannot exist. However, it is still possible that divergent values of the binding free energies lead to an ever decreasing value of $\Psi$ that nonetheless does not tend towards $\Psi_0$. We now argue that $\Psi$ tends exponentially towards zero for sufficiently large $\left(\Delta G_{ES},\Delta G_{EP}\right)$. As a consequence, the procedure of iteratively adding $\delta G>0$ to any candidate pair $\left(\Delta G^{\rm opt}_{ES},\Delta G^{\rm opt}_{EP}\right)$ with $\Psi^{\rm opt} > \Psi_0>0$ is guaranteed to eventually reach an improved solution with reduced $\mathcal{R}$ and $\Psi = \Psi_0$.

In general, the net flux  through the system is 
\begin{align}\label{eq:flux_high_free_energy}
\begin{split}
\Psi &=  k_{+0} \pi_E \times  \text{P}(EP \to E\, \text{occurs before}\, ES \to E) \\ 
   & - k_{-1}\pi_E   \text{P}(ES \to E\, \text{occurs before}\, EP \to E).
\end{split}
\end{align}
This result follows from multiplying the rate of binding transitions in steady state by the probability that those binding transitions actually lead to substrate/product turnover.

In the limit $\Delta G_{ES},\Delta G_{EP} \rightarrow \infty $ at an arbitrary fixed offset, $\Delta G_{EP} = \Delta G_{ES} +\Delta G_{\rm off}$, we have $\pi_E \rightarrow 1$ since both ES and EP states are unstable. 
Moreover, the probability of transitioning between the EP and ES states (rather than to E) is suppressed by a factor  $\exp(-\Delta G_{ES})$. To first order in the small quantity $\exp(-\Delta G_{ES})$, trajectories that visit either EP or  ES more than once before returning to E can therefore be neglected when calculating $\text{P}(EP \to E\, \text{occurs before}\, ES \to E)$ and $ \text{P}(ES \to E\, \text{occurs before}\, EP \to E)$ .

Noting additionally that $\text{rate}(E\to ES)=k_0[S]$ and $\text{rate}(E\to EP) = k_1[P]$, Equation~\ref{eq:flux_high_free_energy} becomes 
\begin{align}
\lim_{\Delta G_{ES} \rightarrow \infty} \Psi = k_0[S] P_{ES \rightarrow EP} - k_1[P] P_{EP \rightarrow ES}.
\end{align}

Moreover, since
\begin{align}
\lim_{\Delta G_{ES} \rightarrow \infty}  P_{ES \rightarrow EP} = \lim_{\Delta G_{ES} \rightarrow \infty}{\rm const} \exp(-\Delta G_{ES}) = 0, & \nonumber \\
\lim_{\Delta G_{ES} \rightarrow \infty}  P_{EP \rightarrow ES} = \lim_{\Delta G_{ES} \rightarrow \infty}{\rm const} \exp(-\Delta G_{ES}) = 0,
\end{align}
$\lim_{\Delta G_{ES} \rightarrow \infty}  {\Psi} = 0$
and the flux constraint is always saturated. \\

\end{proof}

\begin{figure}[htbp]
\centering
\includegraphics[scale=0.7]{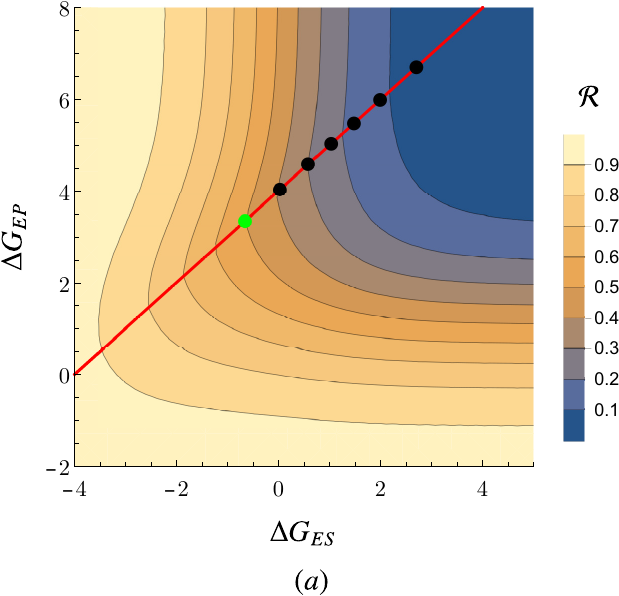}\vspace{3mm}
\includegraphics[scale=0.7]{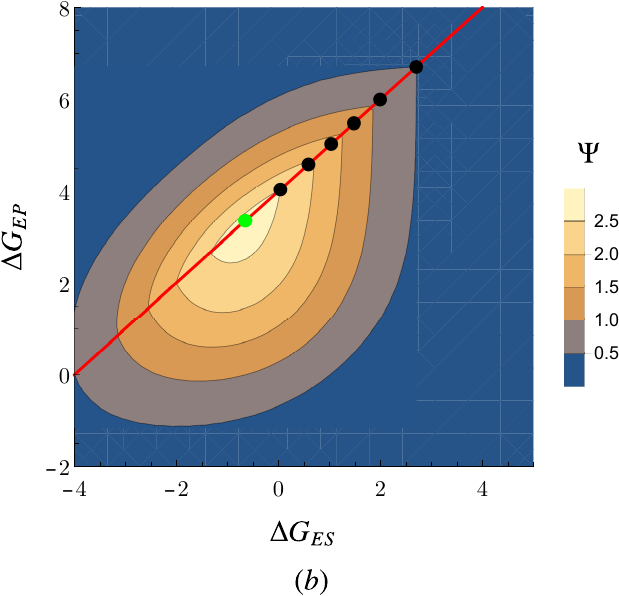}\vspace{5mm}
\includegraphics[scale=0.14]{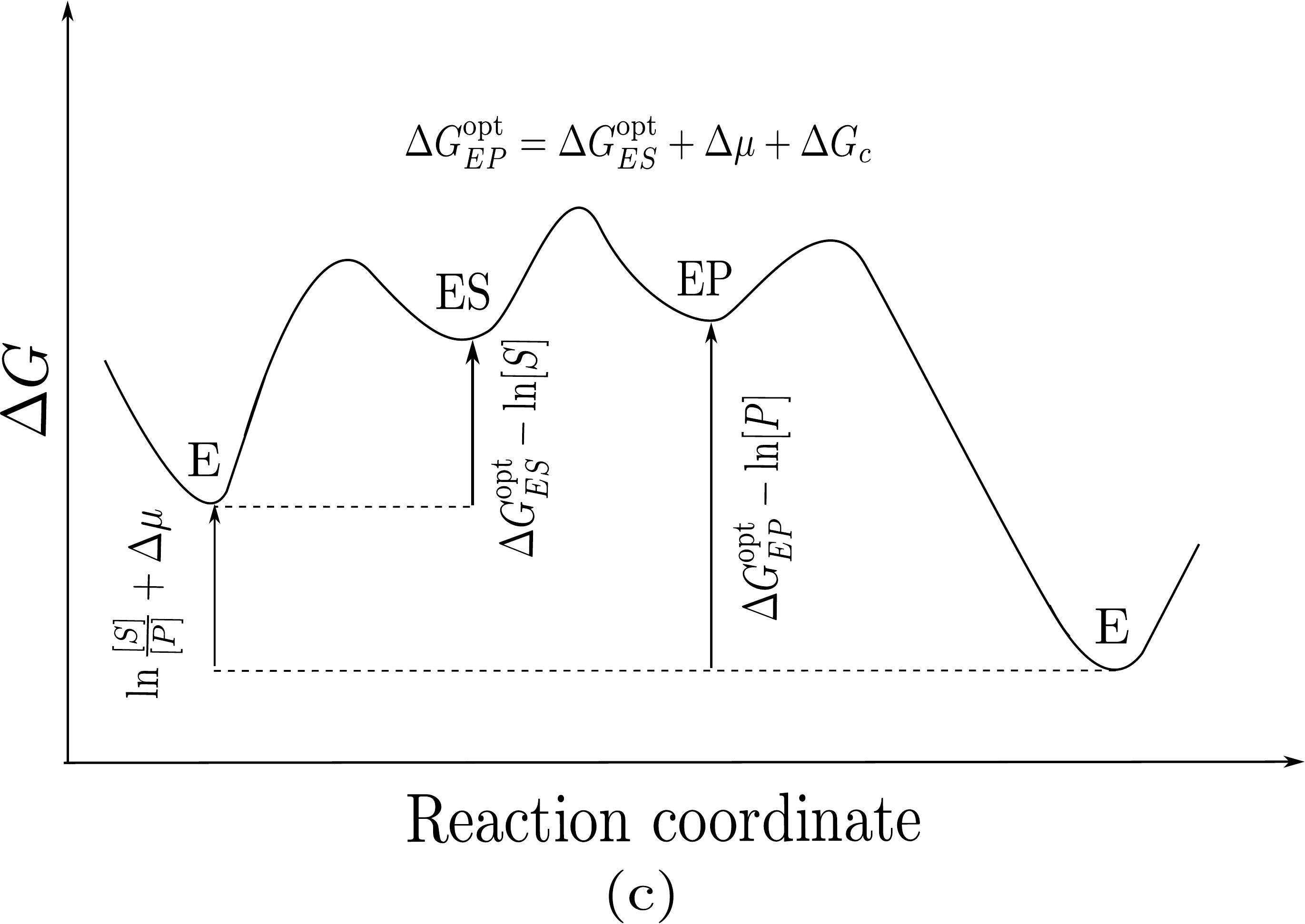}
\captionof{figure}{Numerical calculations with a concrete example, exemplifying the behaviour predicted in Fig.~\ref{fig:diffusion_energy_regions}. (a) Sequestration contour plot as a function of binding free energies. (b) Flux contour as a function of binding fee energies. The black dots in (a) and (b) indicate optimal binding free energies corresponding to certain target fluxes. The green dot is the point corresponding to the maximum flux; all optimal binding free energies lie on the line that corresponds to $\Delta G^{\rm opt}_{EP}=\Delta G^{\rm opt}_{ES}+\Delta\mu + \Delta G_c$ to the right of this point. Parameters used for optimization: $[E_{\rm tot}]=10,[S] = 1,[P] = 3,k_0 = 1,k_1 = 1,k_{\rm cat} = 1,\Delta\mu = 3,\Delta G_c = 1$. Target fluxes $\Psi_0: (0.5,0.9,1.3,1.7,2.1,2.5)$. (c) Illustration of a free energy landscape corresponding to an optimal network with $\Delta G_c \approx 0$. Both ES and EP are moderately high in free energy, and at almost equal height; as a result, the overall free-energy change of reaction $\Delta G= -\ln [S]/[P]-\Delta \mu$ is essentially equal to the difference between the substrate binding free energy ($\Delta G^{\rm opt}_{ES} - \ln[S]$) and the product binding free energy ($\Delta G^{\rm opt}_{EP} - \ln[P]$).}
\label{fig:numerical_result_diffusion}
\end{figure}

\subsubsection{Qualitative physical discussion of diffusion-controlled system}

The general behaviour outlined in Fig.~\ref{fig:diffusion_energy_regions} is exemplified by a specific system in Fig.~\ref{fig:numerical_result_diffusion}, in which we show contour plots of sequestration $\mathcal{R}$ and flux $\Psi$ as a function of $\Delta G_{ES}$ and $\Delta G_{EP}$. Numerical optimization (using the code in 

The flux plot, shown in Fig.~\ref{fig:numerical_result_diffusion}\,b, illustrates the central trade-off inherent to enzymatic operation. Focussing on just the substrate binding free energy $\Delta G_{ES}$, we see that the flux $\Psi$ is non-monotonic, with a peak at moderate values of $\Delta G_{ES}$. If $\Delta G_{ES}$ is too high, $\Psi$ vanishes because ES complexes immediately dissociate before the catalytic reaction can occur. On the other hand, if $\Delta G_{ES}$ is too low, ES complexes are too stable and the reactions proceed to completion extremely slowly. By requiring the system to minimise $\mathcal{R}$, we force the system to the highest possible binding free energies that can sustain the flux, since high free energies tend to reduce binding and therefore sequestration.

This central trade-off is not apparent in the Michaelis-Menten model of Eq.~\ref{eq:MM}. In that model, there is no penalty to the binding free energy of the ES complex being arbitrarily low, because it is assumed that the complex can always be converted to ${\rm E + P}$ quickly. However, in our more complete model, if ES is to be rapidly converted into EP, the EP state must also be low in free energy. If this is the case, however, it will tend to frustrate the subsequent release of the product. As a result, optimal values of  $\Delta G_{ES}$ are moderate.

Beyond this moderation of  $\Delta G^{\rm opt}_{ES}$, we find a linear relationship between $\Delta G^{\rm opt}_{ES}$ and $\Delta G^{\rm opt}_{EP}$: $\Delta G^{\rm opt}_{EP}=\Delta G^{\rm opt}_{ES} + \Delta\mu + \Delta G_c$. In our simple model, this relationship means that reactions can occur arbitrarily fast out of the ES and EP states, avoiding sequestration, without compromising the tendency of the reactions to proceed in the desired direction (${\rm ES \rightarrow EP}$ rather than E, ${\rm EP \rightarrow E}$ rather than ES). At a deeper level, it reflects the fact that there is no point in making the product bind arbitrarily more weakly to the catalyst than the substrate does, or the catalytic step ${\rm ES \rightarrow EP}$ will never occur. Similarly, however, there is no point in driving the ${\rm ES \rightarrow EP}$ reaction forwards using a far more favourable  enzyme-product complex, since this product would never be released. In fact, in the default symmetric case in which the kinetic offset parameter $\Delta G_c=0$ (Fig.~\ref{fig:numerical_result_diffusion}\,c), the difference in standard binding free energies is given by precisely the intrinsic free energy difference of the free product and substrate, $\Delta \mu$. If $\Delta \mu>0$, the released free energy can compensate for a limited increase in $\Delta G^{\rm opt}_{EP}$ relative to $\Delta G^{\rm opt}_{ES}$, allowing a somewhat enhanced rate of product release.

Hitherto, we haven't discussed the free-energy profile for the optimal catalyst as it converts a single substrate into a product in detail. Na\"ively, one might assume that to optimise the rate at which the system moves through its states, the optimal free energies would form a nice ladder of roughly evenly-spaced states. However, in their paper, in which the sole aim was to maximise flux, Brown and Sivak noted that uneven free energy drops could ``compensate for differences in bare rate constants"~\cite{brown2017allocating}. Whilst such an effect is doubtless also present in our system, we also see that the additional need to minimise sequestration leads to the intermediate states being systematically pushed to higher free energies as illustrated in Fig.~\ref{fig:numerical_result_diffusion}\,(c).

\subsection{Non-diffusion controlled binding rates}\label{sec:non_diffusion}

With respect to Equation~\ref{eq:actual_mechanism}, we have assumed hitherto that the binding rate constants $k_{+0},k_{-1}$ are diffusion-controlled and therefore fixed, independent of the binding free energies $\Delta G_{ES}$ and $\Delta G_{EP}$. While this approximation may be reasonable  in many cases, we have effectively assumed that dissociation reactions ${\rm ES} \rightarrow {\rm E+S}$ and ${\rm EP} \rightarrow {\rm E+P}$ and can occur arbitrarily fast if $\Delta G_{ES}$ and $\Delta G_{EP}$ are sufficiently large, which is likely to be unphysical. Fortunately, our solutions predict finite values of $\left(\Delta G^{\rm opt}_{ES},\Delta G^{\rm opt}_{EP}\right)$ for any target flux $\Psi_0$, so our optimal enzymes are not inherently pathological in this fashion. Nonetheless, it is reasonable to consider the possibility that sufficiently large, positive values of $\Delta G_{ES}$ and $\Delta G_{EP}$ cause the assumption of diffusion-controlled binding rates to break down, leading to a ``reaction-limited"  regime in which association rates are suppressed as bonding becomes more and more unfavourable.

\begin{figure*}[htbp]
\centering
\includegraphics[scale=0.35]{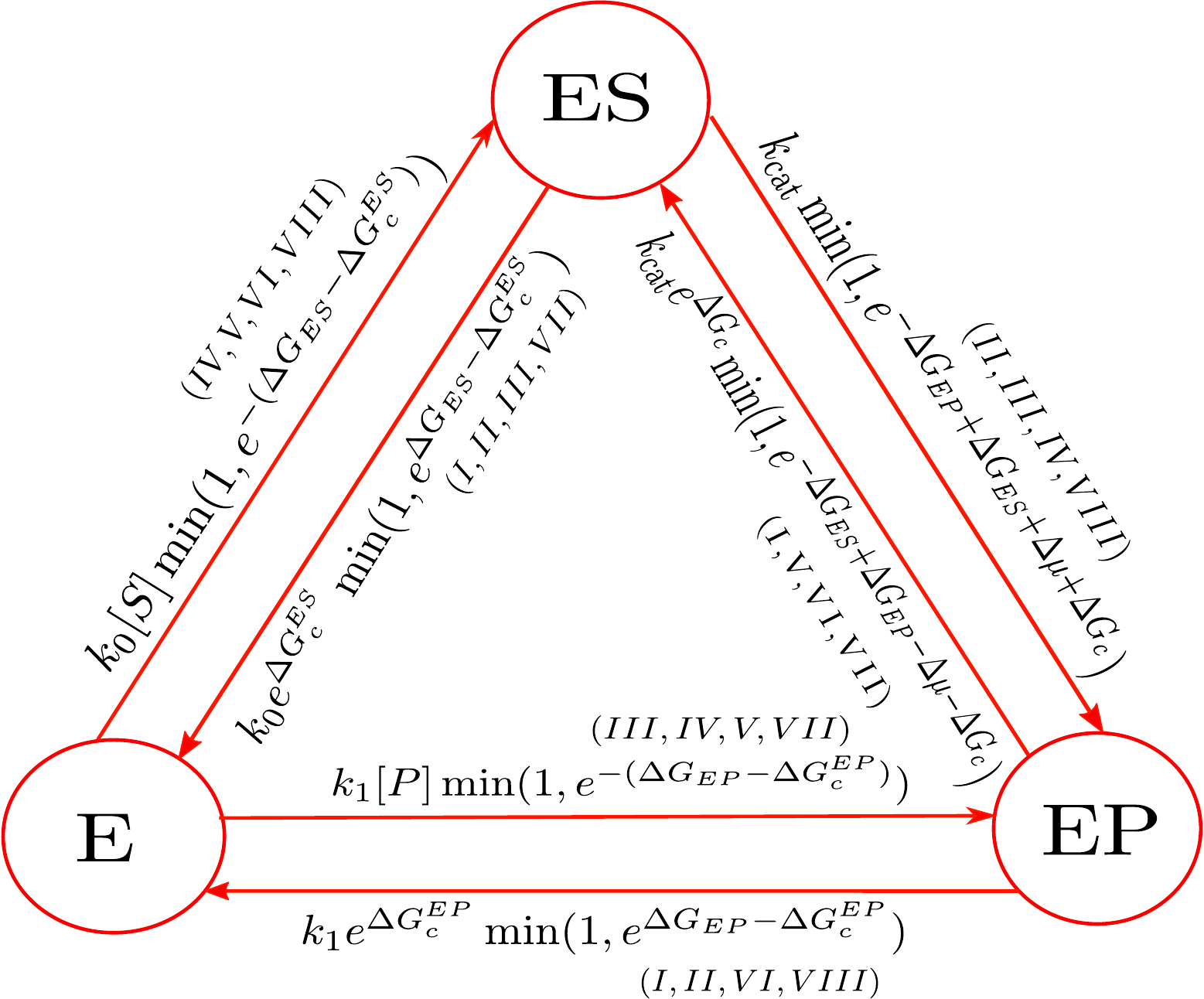}
\captionof{figure}{Markov chain corresponding to Question~\ref{pr:modified} with rate constants governed by Equation~\ref{eq:rates2}. This model incorporates the idea of non-diffusion based binding rates. Within the minimisation statement of a rate constant in the Markov chain, we label one option with a set of regions in the $\Delta G_{ES}-\Delta G_{EP}$ plane in which it applies; the alternative is taken in the other regions. The Markov chain for each of the eight regions can be obtained by applying the appropriate option corresponding to that region in the minimisation statement of the rate constants. 
}
\label{fig:combined_markov_chains}
\end{figure*}

We incorporate this possibility by using the same modified Metropolis dynamics applied to the ${\rm ES \rightleftharpoons EP}$ transition in Section~\ref{sec:diff-controlled}. Specifically, we assume that the ${\rm E+S \rightleftharpoons ES}$ transition is backwards labile (binding is diffusion limited) up until a binding free energy of $\Delta G^{ES}_c$, and forwards labile (reaction limited) above this point. Similarly, we assume that the ${\rm EP \rightleftharpoons E+P}$ transition is forwards labile (binding is diffusion limited) up until a binding free energy of $\Delta G^{EP}_c$, and backwards labile (reaction limited) above this point. Transition rates in the Markov model can then be written

\begin{align}
k_{+0}=k_0[S]\min(1,e^{-(\Delta G_{ES} -\Delta G_c^{ES})}), \nonumber \\
k_{-0}=k_0 e^{\Delta G_c^{ES}}\min(1,e^{\Delta G_{ES}-\Delta G_c^{ES}}), \nonumber \\
k_{+1}=k_1 e^{\Delta G_c^{EP}}\min(1,e^{\Delta G_{EP}-\Delta G_c^{EP}}), \nonumber \\
k_{-1}=k_1[P]\min(1,e^{-(\Delta G_{EP} -\Delta G_c^{EP})}), \nonumber \\
k_{+\rm{cat}}=k_{\text{cat}}\min(1,e^{- \Delta G_{EP}+\Delta G_{ES} + \Delta\mu + \Delta G_c}), & \nonumber\\
k_{-\rm{cat}}=k_{\text{cat}}e^{\Delta G_c}\min(1,e^{-\Delta G_{ES}+ \Delta G_{EP} - \Delta\mu -\Delta G_c}). &
\label{eq:rates2}
\end{align}

The resultant Markov process is illustrated graphically in Fig~\ref{fig:combined_markov_chains}. In this case, there are eight possible combinations of forwards labile and backwards labile options for the three reactions. However, for any particular set of  parameters, only a maximum of seven appear on the $\Delta G_{ES}-\Delta G_{EP}$ plane. This split into seven regions, rather than two as in Section~\ref{sec:diff-controlled}, is illustrated schematically in Figure~\ref{fig:energy_profile}a. We will analyse the resultant mathematical optimization problem in Section~\ref{sec:analysis2}, before turning to its biophysical interpretation in Section~\ref{sec:discussion2}.

\subsubsection{Detailed analysis of non-diffusion-controlled system}\label{sec:analysis2}

Our question in this setting amounts to the following:

\begin{question}\label{pr:modified}
\textit{How should binding free energies $\Delta G_{ES}$ and $\Delta G_{EP}$ in the model of Eq.~\ref{eq:actual_mechanism} be chosen to minimize sequestration $\mathcal{R}$, whilst maintaining a product output rate of $\Psi_0$, given fixed $[S]$, $[P]$ and $\Delta \mu$, non-diffusion-controlled binding reactions, and catalytic steps with fixed and finite upper bounds on their rates given by (Eqs.~\ref{eq:rates2})?} 
\end{question}

\begin{figure*}[htbp]
\centering
\includegraphics[scale=0.087]{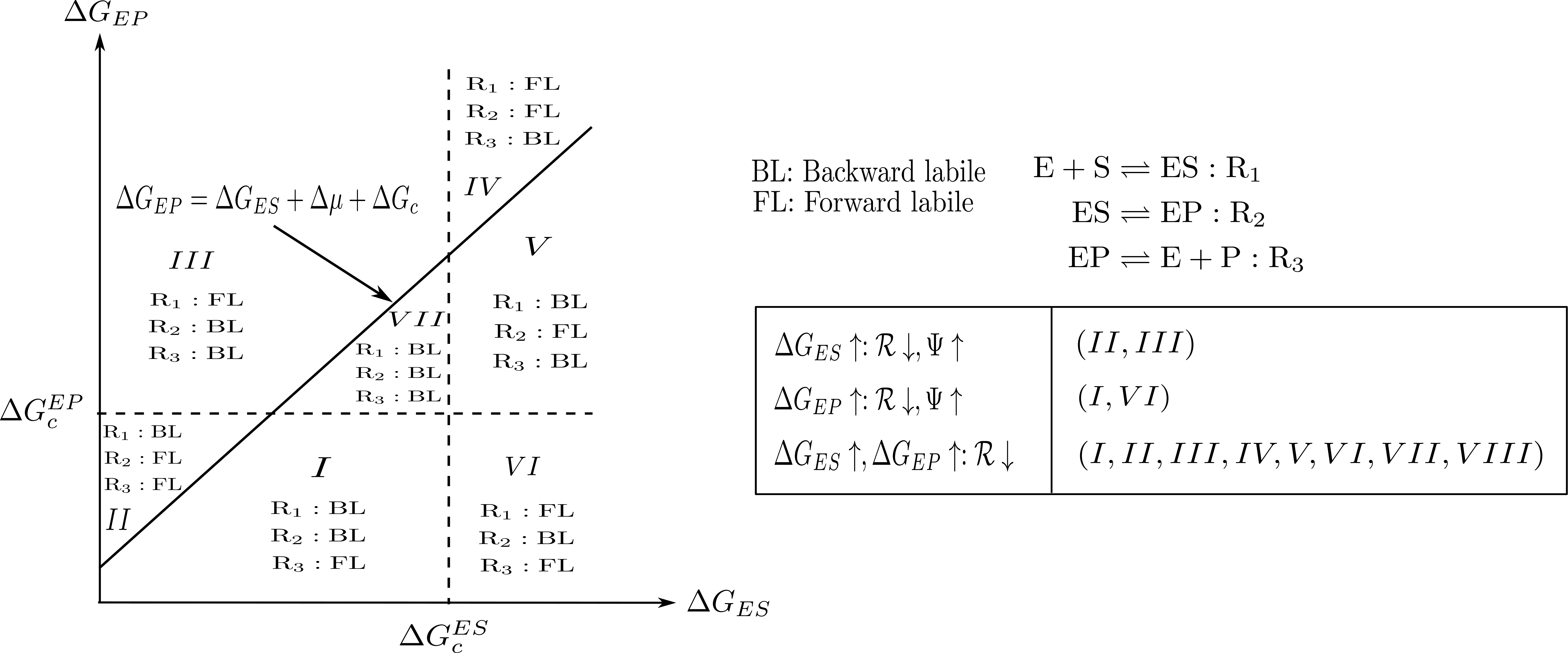}
\captionof{figure}{
Graphical illustration of the optimisation problem~\ref{pr:modified}, in which, at sufficiently high binding free energies, binding is reaction-limited (rather than diffusion-limited). The space of binding energies can be divided into multiple regions in which different reactions are forward and backward labile. If the intersection between $\Delta G_{ES}= \Delta G^{ES}_c$ and $\Delta G_{EP}= \Delta G^{EP}_c$ instead occurs above the line  $\Delta G_{EP}=\Delta G_{ES}+\Delta\mu +\Delta G_c$, region $VII$ would be replaced by region $VIII$, in which all reactions are forwards labile. The table to the right of the figure gives the relationship between retroactivity and flux with respect to the change in binding free energies for various regions in the $\Delta G_{ES}-\Delta G_{EP}$ plane. In regions $II$ and $III$, one can always increase $\Delta G_{ES}$ to get to a state with higher flux and lower sequestration. In regions $I$ and $VI$, one can always increase $\Delta G_{EP}$ to get to a state with higher flux and lower sequestration. Throughout the plane, increasing $\Delta G_{ES}$ and $\Delta G_{EP}$ by the same amount reduces sequestration $\mathcal{R}$; systems with lower target fluxes will therefore be found towards the top right of the graph. Consequently, optimal binding energies lie either on $\Delta G^{\rm opt}_{EP}=\Delta G^{\rm opt}_{ES}+\Delta\mu +\Delta G_c$ or in the regions $IV,V,VII,VIII$.}
\label{fig:energy_profile}
\end{figure*}

The analysis of regions $I$ and $II$ is identical to that in Section \ref{sec:analysis1}. By similar approaches, we deduce the directions on the $\Delta G_{ES}$-$\Delta G_{EP}$ plane that are guaranteed to increase flux and/or decrease sequestration in the different regions shown in Figure~\ref{fig:energy_profile}. 

\begin{itemize}
\item Region $III$: Arguing as we did for Region $II$ during the proof of Lemma~\ref{lem:region_II_behaviour}, one can show that increasing $\Delta G_{ES}$ decreases retroactivity but increases flux in this region.

\item Region $VI$: Employing the same argument that we used for Region $I$ during the proof of Lemma~\ref{lem:region_I_behaviour}, one can show that increasing $\Delta G_{EP}$ decreases retroactivity but increases flux in this region. 

\end{itemize}

We have not found a direction in which $\mathcal{R}$ is guaranteed to decrease, and $\Psi$ guaranteed to increase, for regions $IV$, $V$ $VII$ and $VIII$. As a consequence, we get a weaker result for the augmented system: the optimal binding free energies {\em either} satisfy $\Delta G^{\rm opt}_{EP}=\Delta G^{\rm opt}_{ES}+\Delta\mu +\Delta G_c$ {\em or} lie in regions $IV$, $V$, $VII$ or $VIII$. 

Lemma~\ref{lem:retroactivity_decrease_binding_energies}, however, still holds for the augmented system; if both binding free energies are increased by the same $\delta G >0$, $\mathcal{R}$ necessarily decreases. The result can be separately verified for each region; increasing both binding free energies reduces a non-zero subset of lifetimes $\tau_{ES}$ and $\tau_{EP}$, and transition probabilities $P_{ES \rightarrow EP}$ and $P_{EP \rightarrow ES}$. If adding $\delta G$ to the binding free energies causes the system to move between two regions, the net effects can simply be added together. The result then follows from Eq.~\ref{eq:retroactivity_probability}.

Similarly, we can also show that the flux constraint $\Psi^{\rm opt} \geq \Psi_0$  is saturated for the augmented system. Since, by Lemma~\ref{lem:retroactivity_decrease_binding_energies}, it is always possible to reduce $\mathcal{R}$ by increasing both binding free energies  by the same $\delta G >0$, it only remains to be shown that  $\Psi$ tends exponentially towards zero for sufficiently large $\left(\Delta G_{ES},\Delta G_{EP}\right)$. As a consequence, the procedure of iteratively adding $\delta G>0$ to any candidate pair $\left(\Delta G^{\rm opt}_{ES},\Delta G^{\rm opt}_{EP}\right)$ with $\Psi > \Psi_0>0$ is guaranteed to eventually reach an improved solution with reduced $\mathcal{R}$ and $\Psi = \Psi_0$. 

To perform this analysis, it is necessary to consider regions $IV$ and $V$. Within these regions, consider taking $\Delta G_{ES},\Delta G_{EP} \rightarrow \infty $ at an arbitrary fixed offset,  $\Delta G_{EP} = \Delta G_{ES} +\Delta G_{\rm off}$. Then
\begin{align}
\lim_{\Delta G_{ES} \rightarrow \infty} \Psi = k_0[S] e^{-\Delta G_{ES} +\Delta G_c^{ES}} \frac{P_{ES \rightarrow EP}P_{EP \rightarrow E}}{1-P_{EP \rightarrow ES}P_{ES \rightarrow EP}} & \nonumber \\
-  k_1[P] e^{-\Delta G_{ES} -\Delta G_{\rm off} +\Delta G_c^{ES}} \frac{P_{EP \rightarrow ES}P_{ES \rightarrow E}}{1-P_{ES \rightarrow EP}P_{EP \rightarrow ES}}. 
\end{align}
This result follows from the fact that $\pi_E \rightarrow 1$ as $\Delta G_{ES},\Delta G_{EP} \rightarrow \infty $. Therefore the flux can be calculated as the rate for $E \rightarrow ES$ multiplied by the probability that the transition $EP \rightarrow E$ subsequently occurs before $ES \rightarrow E$, minus the equivalent term for the conversion of $P$ into $S$. 

Within regions $IV$ and $V$, the only effect of increasing $\Delta G_{ES},\Delta G_{EP} $ at a fixed offset is to reduce the rates of the binding transitions by the same factor. All probabilities $P_{i \rightarrow j}$ are unchanged. Therefore
\begin{align}
\lim_{\Delta G_{ES} \rightarrow \infty} \Psi ={\rm const}\, e^{-\Delta G_{ES}}
\end{align}
and the flux necessarily tends exponentially to zero if $\Delta G_{ES},\Delta G_{EP} \rightarrow \infty $ with an arbitrary, fixed offset. We therefore conclude that it is always possible to improve any candidate pair $\left(\Delta G^{\rm opt}_{ES},\Delta G^{\rm opt}_{EP}\right)$ with $\Psi > \Psi_0>0$, and that continued improvements will eventually reach a  solution with reduced $\mathcal{R}$ and $\Psi = \Psi_0$. 

\subsubsection{Qualitative physical discussion of non-diffusion-controlled system}\label{sec:discussion2}

The system with a crossover to reaction-controlled binding kinetics reproduces most of the biophysics observed in the simpler system of Section~\ref{sec:diff-controlled}. Typical examples of $\mathcal{R}$ and $\Psi$ contour plots are given in Fig.~\ref{fig:numerical_result_non_diffusion}, along with points indicating optimal designs for a range of target fluxes $\Psi_0$. The central trade-off that limits $\Delta G_{ES}^{\rm opt}$ to moderate values is still present, and the need to suppress retroactivity pushes binding free energies as high as possible whilst maintaining $\Psi_0$, leading to a free-energy profile that is not uniformly downhill.

\begin{figure*}[htbp]
\centering
\includegraphics[scale=0.7]{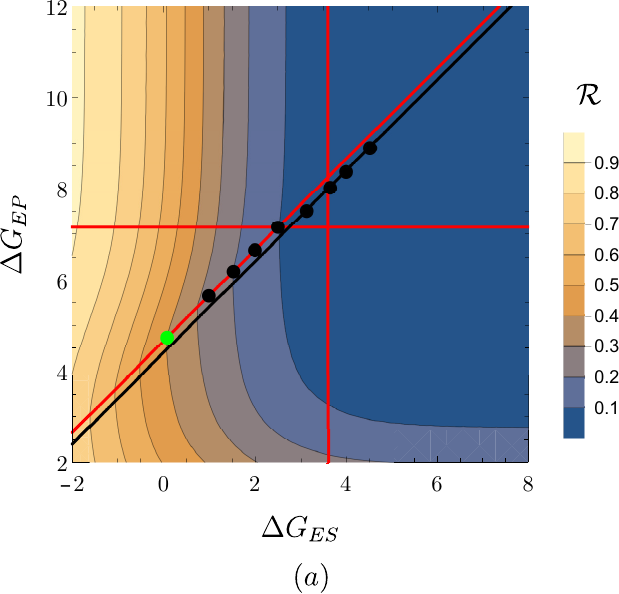}\hspace{9mm}
\includegraphics[scale=0.7]{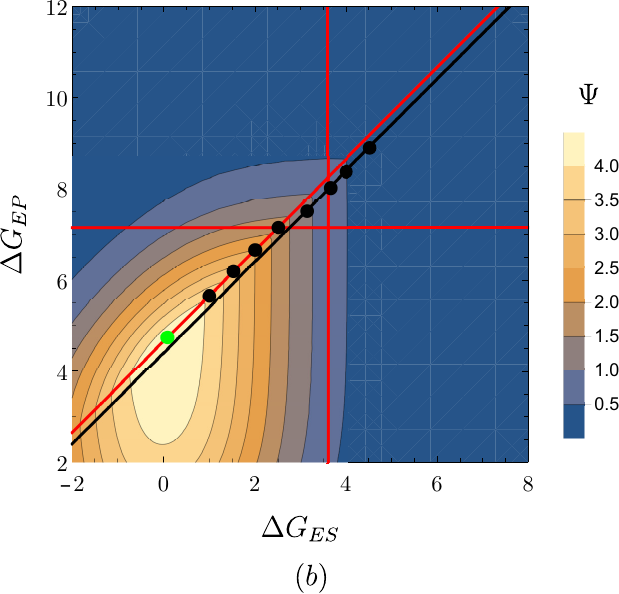}\vspace{10mm}
\includegraphics[scale=0.7]{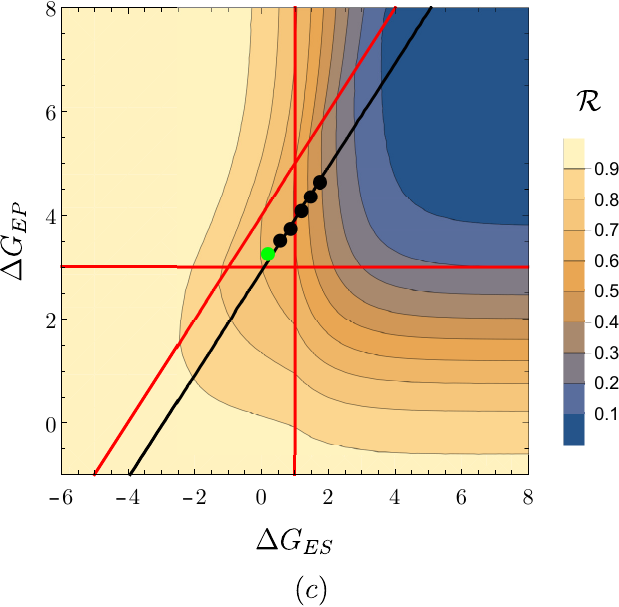}\hspace{9mm}
\includegraphics[scale=0.7]{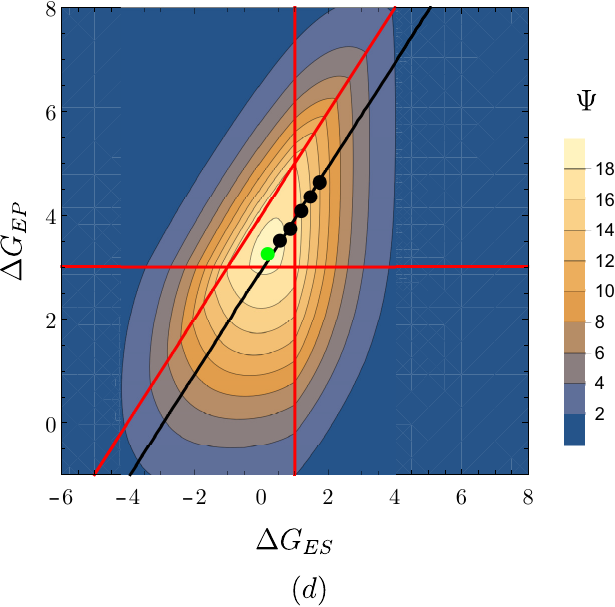}
\captionof{figure}{Optimal solutions for the non-diffusion-controlled system can deviate from the line  $\Delta G^{\rm opt}_{EP}=\Delta G^{\rm opt}_{ES}+\Delta\mu +\Delta G_c$, but they do not deviate that far from this line; they tend towards a parallel line with a constant offset between $\Delta G^{\rm opt}_{EP}$ and $\Delta G^{\rm opt}_{ES}$ black lines with a gradient of unity are drawn as a guide to the eye. We show contour plots for sequestration and flux for two exemplar systems, with optimal binding free energies for different target fluxes $\Psi_0$ shown as points, and the maximal possible flux in the system is illustrated by a green dot. The solid red diagonal line indicates $\Delta G_{EP}=\Delta G_{ES}+\Delta\mu +\Delta G_c$. In the limit of high optimal binding energies, optimal points converge to a line parallel to 
$\Delta G_{EP}=\Delta G_{ES}+\Delta\mu +\Delta G_c$, denoted by a solid balck line in the figure. In (a,b), the peak of flux $\Psi$ is found on the line $\Delta G_{EP}=\Delta G_{ES}+\Delta\mu +\Delta G_c$ between regions $I$ and $II$; some solutions are therefore found on this line to the right of the maximal flux point, but deviation is observed for lower target fluxes which are maximised in regions $V$ and $VII$. In (c,d), the maximal flux is found in region $VII$, and no solutions are found on the line between regions $I$ and $II$. Parameters used for (a,b): $[E_{\rm tot}]=4.87786,[S]=1.61829,[P]=1.76047,k_0=1.62906,k_1=2.80739,k_{\rm cat}=4.61582,\Delta\mu=1.92716, E_1=3.60036,E_2= 7.15195,\Delta G_c=2.72592$. Target fluxes $\Psi_0: (0.3,0.5,0.7,1.1,1.8,2.5,3.2,3.9)$ . Parameters used for (c,d): $[E_{\rm tot}]= 4,[S] = 7,[P] = 5,k_0 = 3,k_1 = 5,k_{\rm cat}=10,\Delta\mu=3,E_1=1,E_2=3,\Delta G_c=1$. Target fluxes $\Psi_0: (12, 14, 16, 17.7, 18.2)$.}
\label{fig:numerical_result_non_diffusion}
\end{figure*}

The major difference is that optimal solutions are no longer constrained to lie exactly on the line $\Delta G^{\rm opt}_{EP}=\Delta G^{\rm opt}_{ES} + \Delta\mu + \Delta G_c$ if the free energies are sufficiently positive. To understand why, consider what happens when we decrease $\Delta G_{ES}$ and move away from the line $\Delta G^{\rm opt}_{EP}=\Delta G^{\rm opt}_{ES} + \Delta\mu + \Delta G_c$ into region $IV$. Generally, decreasing binding free energies increases sequestration. However, unlike in region $II$, such a move can potentially increase the flux $\Psi$, because the rate of substrate binding now increases exponentially. Against this fact, the rate of transition for ${\rm ES \rightarrow EP}$ will also be reduced, tending to reduce $\Psi$. Moving away from the line is therefore now a potentially fruitful way to trade sequestration for flux if the increase in flux due to the binding rate outweighs the negative contribution due to the decrease in the catalytic rate.

Qualitatively, the relative strength of the two contributions depends on whether this reduction in the rate of ${\rm ES \rightarrow EP}$ has a noticeable effect on the probability to proceed to the EP state once bound in the ES state: 
\begin{align}
P_{ES \rightarrow EP} = \frac{k_{\text{cat}}e^{- \Delta G_{EP}+\Delta G_{ES} + \Delta\mu + \Delta G_c}}{k_{\text{cat}}e^{- \Delta G_{EP}+\Delta G_{ES} + \Delta\mu + \Delta G_c}+k_0 e^{\Delta G_c^{ES}}}.
\label{eq:PES->EP}
\end{align}
If the first term in the denominator of Eq.~\ref{eq:PES->EP} is larger the second, then transitions ${\rm ES \rightarrow EP}$  are sufficiently rapid relative to ${\rm ES \rightarrow E}$ that decreasing $\Delta G_{ES}$ by a certain amount affects $P_{ES \rightarrow EP}$, and hence the flux, by much less than the exponential increase in the binding rate. Thus leaving the $\Delta G^{\rm opt}_{EP}=\Delta G^{\rm opt}_{ES} + \Delta\mu + \Delta G_c$ line to enter region $IV$ can be beneficial.

This strategy, however, is clearly limited. A significant decrease in $\Delta G_{ES}$ will quickly suppress the catalytic step, and $P_{ES \rightarrow EP}$ will quickly tend towards $P_{ES \rightarrow EP} \propto \exp^{\Delta G_{ES}}$. At this point, there will be negligible gains in the flux by continuing to reduce $\Delta G_{ES}$, while increases in $\mathcal{R}$ will continue. We thus expect any deviation from the line $\Delta G^{\rm opt}_{EP}=\Delta G^{\rm opt}_{ES} + \Delta\mu + \Delta G_c$ into region $IV$  to be limited.

A similar explanation can be made for optimal solutions in region $V$; in this case, a decrease in $\Delta G_{EP}$ will tend to increase sequestration, but unlike in region $I$ this decrease will tend to slow down unwanted backwards steps ${\rm EP \rightarrow ES}$ without also slowing the release of product from the $EP$ state by the same amount. This tactic is potentially effective in increasing the flux $\Psi$ only while the probability of backwards steps,
 \begin{align}
P_{EP \rightarrow ES} = \frac{k_{\text{cat}}e^{\Delta G_{EP}-\Delta G_{ES} -\Delta\mu}}{k_{\text{cat}}e^{ \Delta G_{EP}-\Delta G_{ES} -\Delta\mu }+k_1 e^{\Delta G_c^{EP}}},
\label{eq:PEP->ES}
\end{align}
remains high. Again, however, the exponential reduction in the rate of backwards steps with $\Delta G_{EP}$ will quickly reduce the incentive to decrease $\Delta G_{EP}$ further, whereas the sequestration of the enzyme by the product will continue to increase. We thus also expect  any deviation from the line $\Delta G^{\rm opt}_{EP}=\Delta G^{\rm opt}_{ES} + \Delta\mu + \Delta G_c$ into region $V$  to be limited.

Despite the breakdown of the tight constraint on $\Delta G^{\rm opt}_{EP}$ and $\Delta G^{\rm opt}_{ES}$, we therefore still expect the free energies to be closely linked. Moreover, this link is due to the physical principles outlined in Section~\ref{sec:diff-controlled}: $\Delta G^{\rm opt}_{ES}$ shouldn't be reduced too far relative to $\Delta G^{\rm opt}_{EP}$, or it will compromise the ${\rm ES \rightarrow EP}$ transition; and $\Delta G^{\rm opt}_{EP}$ shouldn't be reduced too far relative to $\Delta G^{\rm opt}_{ES}$ or any favourable increase in the tendency for  ${\rm ES \rightarrow EP}$  to be unidirectional will be outweighed by increased product binding.  Randomly-generated example systems bear out the intuition of this semi-quantitative analysis (Fig.~\ref{fig:numerical_result_non_diffusion}). Optimal points stay relatively close to the  $\Delta G^{\rm opt}_{EP}=\Delta G^{\rm opt}_{ES} + \Delta\mu + \Delta G_c$ line, and importantly do not appear to continue to move further and further away from it as the target flux $\Psi_0$ is taken to zero. Indeed, optimal pairs $\left(\Delta G^{\rm opt}_{EP},\Delta G^{\rm opt}_{ES}\right)$ tend toward a line parallel to $\Delta G^{\rm opt}_{EP}=\Delta G^{\rm opt}_{ES} + \Delta\mu + \Delta G_c$, as would be predicted if the deviation from the line was constrained by a system-specific limit on the fractions in Eq.~\ref{eq:PES->EP} and \ref{eq:PEP->ES}. 

\section{Conclusions}
We have addressed the question of how catalytic enzymes might be designed or evolved to achieve a target rate of substrate turnover whilst minimizing enzymatic sequestration. This goal maximises the efficiency of metabolic enzymes, and is also helpful in the context of signalling and information-processing tasks. Our results, highlighting key trade-offs in the structure of intermediate enzyme-substrate complexes, should inform the way we rationalize the operation of natural systems and design synthetic analogs. These results are particularly relevant to the design of synthetic information-processing networks -- eg. \cite{Lankinen2020} -- since sequestration-driven retroactivity \cite{tuanase2006signal,ventura2008hidden,del2008modular,barton2013energy,del2015biomolecular,deshpande2017high} disrupts our ability to rationally design complex circuits from simple modular components.

Specifically, we have demonstrated that investigating this question meaningfully requires the use of a more sophisticated model than the standard Michaelis-Menten description of enzymatic kinetics. Using a three-state model, with physically reasonable assumptions on the dependence of reaction rates on binding free energies, we have shown that this challenge centres around the key trade-off in enzymatic kinetics. Namely, binding to substrates should neither be too strong, since the substrate-bound states will act as stable sinks, nor should it be too weak, since the progress of the reaction will be hampered. The optimal binding free energies are therefore moderate.

We also find that the optimal binding free energies of the substrate and the product are closely related to each other. In our simplest description, their difference $\Delta G^{\rm opt}_{ES}-\Delta G^{\rm opt}_{EP}$ is a constant, related to the intrinsic free-energy difference between substrates and products. This design principle runs counter to the intuition that either: (a) $\Delta G^{\rm opt}_{ES}-\Delta G^{\rm opt}_{EP}$ should be extremely negative to favour release of the product; or (b) $\Delta G^{\rm opt}_{ES}-\Delta G^{\rm opt}_{EP}$ should be very positive to favour the enzyme-bound catalytic turnover. In fact, to maximise flux at fixed sequestration, an intermediate value is preferred, with the intrinsic chemical free energy difference between P and S potentially compensating for a slightly negative $\Delta G^{\rm opt}_{ES}-\Delta G^{\rm opt}_{EP}$.

Interestingly, our model predicts that the difference $\Delta G^{\rm opt}_{ES}-\Delta G^{\rm opt}_{EP}$ (unlike the individual optimal binding free energies) is largely insensitive to $[P]$ and $[S]$, and is an intrinsic feature of the enzyme-substrate-product chemistry. It would therefore be a natural candidate for an experimental or bioinformatics investigation exploring the question of which, if any, natural enzymes are optimized to minimize sequestration at fixed flux. It should be noted that even for biological systems that would benefit from reduced sequestration, other constraints are present in a real environment. For example, we have not considered how minimising sequestration might affect the specificity with which an enzyme interacts with one substrate rather than alternatives. 

One way to phrase the above result is that the free energy released by product formation should be exploited to weaken the binding between the product and the catalyst to optimize turnover while minimising sequestration. Although our work has focused on a simple catalytic conversion of a single substrate into a product, this idea is more general. Another important mechanism of cellular information processing is the copying of template polymer sequences, as in RNA transcription and protein translation \cite{crick1970central}. The templates involved must act catalytically to fulfill their biological function \cite{ouldridge_2017_fundamental, poulton2020}, but current attempts to emulate these systems in synthetic contexts have struggled to generate spontaneous detachment of the products \cite{Zhou2019}; the templates effectively become sequestered by the copy polymers formed. Our results here suggest that an effective strategy would be to channel the free energy released by polymerisation into weakening interactions between template and copy, and we have recently presented a DNA reaction motif with exactly the required functionality \cite{CabelloGarcia2020}.

\teo{In Ref.~\cite{brown2017allocating}, Brown and Sivak considered the challenge of maximizing the flux of trajectories through a 3-state system; they did not simultaneously consider minimising the occupancy of the intermediate states, as we do here. They found that the ideal division of free-energy drops was dependent on the detailed functional dependence of the transition rates on free-energy differences.} To meaningfully tackle our optimisation problem, it was necessary not only to make some clear assumptions about these functional forms. We had to go beyond the simple characterisation of reactions as ``forwards labile" and ``backwards labile"~\cite{brown2017allocating}, by introducing limits to the speed of any one reaction --   an approach that is likely to be relevant to a broad class of optimization problems in molecular systems. 

\teo{By introducing these constraints, we found tighter restrictions on the shape of free energy landscapes than Brown and Sivak - in particular, the tendency of the free energy difference between intermediate states to be closely related to the overall free energy change of reaction. Moreover, the challenge of optimizing flux alone would give only the single point of maximal flux highlighted in green in Figs.~\ref{fig:numerical_result_diffusion}, \ref{fig:numerical_result_non_diffusion}. 
Demanding a certain flux with minimal sequestration instead gives a family of solutions, with lower required fluxes leading to higher binding free energies for both enzyme-substrate and enzyme-product complexes. These families of optimal solutions tend to follow lines with gradient 1 in $\Delta G_{\rm ES} - \Delta G_{\rm EP}$-space, illustrating clearly the close relationship between enzyme-substrate and enzyme-product binding free energies in the context.}

\teo{Our model of a catalyst with two discrete, metastable intermediate states is highly simplified, and more sophisticated descriptions with a higher number or even a continuum of intermediate states are possible. More complex descriptions of the dependence of transition rates on system parameters could also be considered, such as a smoothed version of the current sharp transition between forwards and backwards labile as the free energy of reaction is changed. A candidate would be transition rates proportional to $1-\tanh (\Delta G/2)$, where $\Delta G$ is the free energy change assoicated with the transition. Our proofs, which rely on  the specific functional forms of the rates, may not hold exactly in such a smoothed model. However, we expect that smoothed models would produce results that are consistent with the overall biophysics identified here, with the smoothing region making otherwise exact results approximate.} It wold also be interesting to consider time-dependent concentrations of substrate and product (rather than a fixed representative concentration for each). As a general point, however, we have argued at our model is the simplest in which the question of minimizing sequestration at fixed flux has a meaningful answer, and our approach therefore provides insight into the problem in general. Any additional detail is probably best suited to the study of specific systems in which the myriad choices necessary can be justified by reference to the particular biochemistry.

It would perhaps be of greater interest to treat the turnover of  any ancillary fuel molecules, such as ATP, explicitly. In particular, we have ignored the possibility of futile cycles in which fuel molecules are consumed without converting substrate into product. It is possible that such cycles might allow reduced retroactivity at a given flux, at the expense of addition consumption of chemical free energy, analogous to the increases in molecular recognition specificity that can be achieved through ``kinetic proofreading" schemes~\cite{hopfield1974kinetic,ninio1975kinetic,mckeithan1995kinetic}. It would be instructive to compare such a strategy to the effect of simply increasing the intrinsic chemical free energy difference between substrates and products. Previous work~\cite{deshpande2017high,barton2013energy} has shown that fuel consuming approaches can be effective at reducing retroactivity when optimising at the level of the molecular network; we raise the question of whether this strategy is also effective at the level of individual enzyme-substrate-product interactions.

\section{Codes}
The codes for generating the figures are available at 

\section{Acknowledgement}

A.D. acknowledges support from Roth scholarship during his time at the Department of Mathematics in Imperial College London and Van Vleck Visiting Assistant Professorship from the Department of Mathematics at Wisconsin Madison. T.E.O. is supported by a Royal Society University Research Fellowship.

\bibliographystyle{unsrt}
\bibliography{Bibliography}

\end{document}